\documentclass[]{article}
\usepackage[left=2cm,right=2cm, top=2.4cm,bottom=2.4cm,bindingoffset=0cm]{geometry}
\usepackage{tikz, titlesec, tikz-cd}
\usetikzlibrary{calc,intersections,through,backgrounds}
\usepackage{amsmath, amssymb}
\usepackage{amsthm, mathtools}

\titleformat*{\section}{\large\bfseries}
\titleformat*{\subsection}{\normalsize\bfseries}

\usepackage{array}
\newcolumntype{?}{!{\vrule width 1pt}}

\usepackage{ifthen}

\usepackage{faktor}

\mathtoolsset{showonlyrefs}

\usepackage{makecell}

\usepackage[numbers, sort&compress]{natbib}

\usepackage[bottom]{footmisc}

\usepackage{hyperref}

\definecolor{color1}{RGB}{127,201,127}
\definecolor{color2}{RGB}{190,174,212}
\definecolor{color3}{RGB}{253,192,134}
\definecolor{color4}{RGB}{255,255,153}

\usetikzlibrary{positioning}
\usetikzlibrary{decorations.text}
\usetikzlibrary{decorations.pathmorphing}
\usetikzlibrary{decorations.markings}

\usepackage{tkz-euclide}

\newtheorem{lemma}{Lemma}[section]
\newtheorem{proposition}[lemma]{Proposition}
\newtheorem{theorem}[lemma]{Theorem}
\newtheorem{corollary}[lemma]{Corollary}

\theoremstyle{definition}
\newtheorem{remark}[lemma]{Remark}
\newtheorem{definition}[lemma]{Definition}

\newcommand{\Z}{\mathbb{Z}}
\newcommand{\C}{\mathbb{C}}
\newcommand{\R}{\mathbb{R}}
\renewcommand{\H}{\mathbb{H}}

\newcommand{\Q}{\mathcal{Q}}

\renewcommand{\i}{\mathbf{i}}
\renewcommand{\j}{\mathbf{j}}
\renewcommand{\k}{\mathbf{k}}
\newcommand{\D}{\mathcal{D}}
\usepackage{bbm}

\title{Polygon recutting as a cluster integrable system}
\author{Anton Izosimov\thanks{
Department of Mathematics,
University of Arizona;
e-mail: {\tt izosimov@math.arizona.edu}
} }
\date{}
\begin{document}

\tikzset{->-/.style={decoration={
  markings,
  mark=at position .7 with {\arrow{>}}},postaction={decorate}}}
  
  \tikzset{-<-/.style={decoration={
  markings,
  mark=at position .3 with {\arrow{<}}},postaction={decorate}}}
  
\usetikzlibrary{angles, quotes}

\maketitle

\abstract{Recutting is an operation on planar polygons defined by cutting a polygon along a
 diagonal to remove a triangle, and then reattaching the triangle along the
same diagonal but with opposite orientation.  Recuttings along different diagonals generate an action of the affine symmetric group on the space of polygons. We show that this action is given by cluster transformations and is completely integrable. The integrability proof is based on interpretation of recutting as refactorization of quaternionic polynomials.

%The corresponding quiver can be embedded on a torus, from which we deduce complete integrability of the recutting dynamics. 

%Performing consecutive cuts along arbitrarily chosen diagonals, one obtains a non-trivial dynamical system. In the present paper we identify polygon recutting with mutations in a cluster algebra. The corresponding quiver can be embedded on a torus, from which we deduce complete integrability of the recutting dynamics. 

%Adler showed that cuts performed along different diagonals generate an action of the affine Weyl group of type $A$.  Furthermore, he constructed a large number of invariants for that action, which suggests the action is completely integrable. \par
%In the present paper we show that polygon recutting can be identified with a sequence of cluster mutations for an appropriate quiver. Furthermore, that quiver can be drawn on a torus, from which we deduce that polygon recutting is a completely integrable system.
}

\tableofcontents

\section{Introduction}
%\subsection{Background and main results}
%\paragraph{Background.} 
Recent years have seen a spark of interest in discrete integrable systems, largely due to emerging connections with cluster algebras. Many of such systems are defined by iterating a certain geometric construction, with Schwartz's \textit{pentagram map}  \cite{OST} being the best known example. % (Google Scholar knows of 207 papers mentioning the pentagram map, so we do not even attempt to give a complete list of references). %\footnote{The pentagram map was introduced in~\cite{Sch}. There is presently an enormous amount of literature on the subject, and its review is beyond the scope of the present paper.} being the best known example. %The literature on the pentagram map is abundant, so we just mention a few works most relevant to the present paper. In \cite{OST}, it is proved that the pentagram map is a completely integrable system. 
In the present paper we study another dynamical system of a somewhat similar nature: Adler's \textit{polygon recutting}~\cite{Adler}. %Just like the pentagram map, recutting uses diagonals of a polygon to build a new polygon. However, instead of intersecting diagonals, one  cuts along a diagonal. Specifically, 
Given a planar polygon, its recutting $\rho_i$ at a vertex $v_i$ is defined as follows. Detach the triangle formed by the vertices $v_{i-1}, v_i, v_{i+1}$ from the rest of the polygon by cutting along the diagonal $v_{i-1}v_{i+1}$. Then attach the triangle back  along the
same diagonal but with opposite orientation. Put differently, recutting $\rho_i$ at a vertex $v_i$ is reflection of $v_i$ in the perpendicular bisector of the diagonal $v_{i-1}v_{i+1}$, see Figure \ref{Fig2}. 

\begin{figure}[t]
\centering
\begin{tikzpicture}[scale = 1.1]
\coordinate (A) at (1,-0.5) {};
\coordinate[label=left:{$ v_{i-1}$}] (B) at (0,1) {};
\coordinate [label={[label distance=0pt]90:$v_i$}] (C) at (1,1.8) {};
%\node [draw,circle,color=blue, fill=blue,inner sep=0pt,minimum size=2pt] () at (C) {};
\coordinate [label={[label distance=0pt]90:$v_i'$}] (Cp) at (2,1.8) {};
%\node [draw,circle,color=red, fill=red,inner sep=0pt,minimum size=2pt] () at (Cp) {};
\coordinate [label=right:{$ v_{i+1}$}] (D) at (3,1) {};
\coordinate (E) at (2.5,0) {};
%\coordinate (G) at (1, -0.7) {};
%\coordinate (F) at (2, -0.7) {};
%\coordinate (F) at (1.5,1) {};
\coordinate (G) at (1.5,2) {};
\coordinate (H) at (1.5,0.5) {};

%\node  () at (1.5,-0.5) {\dots};
\draw [] (B) -- (D);
\draw [] (H) -- (G);
\draw [thick, red](B) -- (Cp) -- (D) ;
%\draw [ fill = red,  opacity = 0.1] (A) -- (B) -- (Cp) -- (D) -- (E)  -- cycle;
\draw [blue] (A) -- (B) -- (C) -- (D) -- (E)  -- cycle;
\end{tikzpicture}
\caption{Recutting $\rho_i$ at the vertex $v_i$ moves it to position $v_i'$. Other vertices remain intact.  }\label{Fig2}
\end{figure}
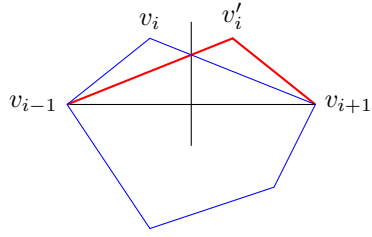
\par

%Along with individual recuttings $\rho_i$, we will consider the \textit{recutting group} $G_R = \langle \rho_1, \dots, \rho_n \rangle $ 
%The goal of the present paper is to understand the dynamics defined by compositions of recuttings  $\rho_i$ at different vertices taken in arbitrary order. It has long been known that this dynamics possesses many features of an integrable system:

Consider the group generated by recuttings $\rho_i$ at arbitrary vertices. The goal of the present paper is to understand the dynamics of that group as it acts on the space of polygons. It has long been known that this dynamics possesses many features of an integrable system:
\begin{itemize}
\item As observed in \cite{Adler} recuttings $\rho_i$ of a closed $n$-gon obey the relations of the affine symmetric group~$\tilde S_n$ (also known as the affine Weyl group $\tilde A_{n-1}$). Specifically, one has \begin{enumerate} \item $\rho_i^2 = \mathrm{id}$ for any vertex~$v_i$; \item $\rho_i \rho_j = \rho_j \rho_i$ for any non-consecutive vertices $v_i, v_j$ (here indices are considered modulo $n$, so that the vertices $v_n$ and $v_1$ are thought of as consecutive); \item the \textit{braid relation}
$
\rho_i \rho_{i+1}\rho_i = \rho_{i+1} \rho_{i}\rho_{i+1}
$ 
for any consecutive (modulo $n$) vertices $v_i$, $v_{i+1}$ (see  \cite{Adler2} and Proposition \ref{prop:braid} below).
\end{enumerate}
%The first two relations are geometrically obvious. The braid relation is much less so. We give a geometric proof of the braid relation in the appendix. 

These relations in particular imply that the group generated by recuttings $\rho_i$ has {polynomial growth}. According to \cite{Veselov}, this is a necessary condition for integrability of a group action. In what follows, we refer to the action of  the affine symmetric group~$\tilde S_n$ on $n$-gons by recutting as the \textit{recutting action}, and to $\tilde S_n$ itself as the \textit{recutting group}.

\item Another result of  \cite{Adler}  is that recuttings constitute discrete symmetries of an integrable system known as the  \textit{dressing chain} \cite{veselov1993dressing}. This in particular provides a Lax (or zero curvature) representation and a number of invariants (first integrals) for  recutting dynamics. 
\item A different Lax representation (which also applies to recutting of polygons in spaces of dimension $d > 2$) is given in \cite{Adler2}.
\item The paper \cite{reshetikhin2005poisson} provides a Poisson-Lie group model for a discrete system which can be thought of as an extension of recutting. %This gives another way to construct invariants, although is is not very clear whether the Poisson structure on the group can be restricted to 
\item As shown in \cite{tabachnikov2012discrete}, recutting commutes with another conjecturally integrable system, the so-called \textit{discrete bicycle transformation}. The latter has a large number of invariants which are also preserved by recutting dynamics (see Remark \ref{rem:byc} below).
\item The paper \cite{glick2015devron} shows that recutting has the so-called \textit{Devron property}, a highly structured behavior of singularities common for cluster integrable systems.
\end{itemize}

 The main result of the present paper is that recutting of planar polygons is {indeed a completely integrable system}. Moreover, it is a \textit{cluster integrable system}, meaning that recutting at any vertex is a ($Y$-type) cluster transformation, and recutting invariants (first integrals) commute with respect to the log-canonical Poisson bracket associated with the corresponding quiver.

 \begin{table}[b!]
\centering
\begin{tabular}{?c?c|c|c?} \Xhline{2\arrayrulewidth} Space & \parbox{4cm}{\centering $\mathcal P_n^S  / S $, planar polygons closed up to similarity modulo similarities\\  } &  \parbox{4cm}{\centering \vskip 0.2cm $\mathcal P_n^E  / E$, planar polygons closed up to isometry modulo isometries \\ \vskip 0.2cm}  &  \parbox{4cm}{\centering $\mathcal P_n  / E$, closed planar polygons modulo isometries} \\\Xhline{2\arrayrulewidth}  %\parbox{3cm}{\centering{Parametrized by\\}} &  \parbox{4cm}{\centering Cluster variables} &  \parbox{4cm}{\centering  \vskip 0.2cm  Network edge weights modulo gauge transformations\\ \vskip 0.2cm } & - \\\hline 
\parbox{3cm}{\centering{\quad \\Recutting-invariant Poisson structure\\ \vskip 0.3cm}} &     \parbox{4cm}{\centering Cluster Poisson structure} &   \parbox{4cm}{\centering Poisson structure on quaternionic polynomials}  & \parbox{4cm}{\centering -} \\\hline  \parbox{3cm}{\centering{\quad\\Recutting invariants (first integrals)\\\vskip 0.2cm}} & \parbox{4cm}{\centering -} &  \parbox{4cm}{\centering Yes} &  \parbox{4cm}{\centering Yes}  \\\hline   \parbox{3cm}{\centering{ \vskip 0.2cm  Recutting is ...\\ \vskip 0.2cm }} & \parbox{4cm}{\centering ... given by cluster transformations (Proposition \ref{prop:cluster})} &  \parbox{4cm}{\centering \vskip 0.2cm ... Arnold-Liouville integrable (Theorem \ref{thm1}) \vskip 0.2cm} & \parbox{4cm}{\centering \vskip 0.2cm ... integrable in the non-Hamiltonian sense (Theorem \ref{thm2}) \vskip 0.2cm}  \\   \Xhline{2\arrayrulewidth}\end{tabular}
\caption{ Polygon spaces and associated structures.}\label{table1}
\end{table}

We summarize our results in Table \ref{table1}. As can be seen from the table, different structures that arise in connection with recutting are defined on spaces of polygons with different periodicity conditions. The largest space that we consider is planar polygons closed up to an orientation-preserving similarity (i.e. a composition of rotations, translations, and homotheties). Such a polygon is understood as a bi-infinite sequence $v_i \in \C$ satisfying a quasi-periodicity condition $p_{i+n} = \psi(v_i)$ where $\psi$ is a similarity transformation $z \mapsto az + b$, called the \textit{monodromy} of the polygon. %This space can be regarded as one of the possible analogues  of the space of \textit{twisted polygons} in the projective plane
%As natural coordinates on that space one can take ratios $y_i := z_{i+1}/z_i$ where $z_i := v_i - v_{i-1}$ are edge vectors of the polygon. The condition of closedness of the polygon $v_i$ up to similarity is equivalent to periodicity of the corresponding sequence $y_i$. %Our main result conce
%\begin{theorem}
%Polygon recutting on the space of $n$-sons closed up to similarity is given by $Y$-type cluster mutations of the variables $y_i$ associated with a certain quiver $\mathcal Q_n$.
%\end{theorem}
Let $\mathcal P_n^S \, / S$ be the space of planar $n$-gons closed up to similarity, modulo similarities. 
In Section \ref{sec:cluster} we interpret recutting on that space in terms of $Y$-type cluster mutations of a certain quiver $\Q_n$. Figure~\ref{Fig3} shows the quiver $\Q_5$ corresponding to pentagons. The general quiver $\Q_n$ has a similar structure, but with $2n$ vertices. Geometrically, the variables $y_i$ are the ratios of consecutive edges of the polygon, viewed as complex numbers, while $\bar y_i$ are complex conjugates of $y_i$. Recutting $\rho_i$ is achieved by mutation of quiver vertices $y_i$, $\bar y_i$ followed by interchanging those vertices.  This sequence of mutations is an example of a more general transformation known as a \textit{geometric R-matrix}~\cite{inoue2019cluster}. Geometric $R$-matrices are known to satisfy braid relations, which gives yet another proof of the braid relation for recutting. %{\color{red}It is well known that geometric $R$-matrices generate an affine symmetric group action, so their appearance in connection with recutting does not come as a surprise.}%experts will not be surprised to learn about their connection with recutting.\par

\begin{figure}[t]
\centering
\begin{tikzpicture}[
  ->,   
  thick,
  main node/.style={},
scale = 0.7]
  \newcommand*{\MainNum}{5}
  \newcommand*{\MainRadius}{3cm} 
    \newcommand*{\SecondRadius}{1.5cm} 
  \newcommand*{\MainStartAngle}{90}

  % Print main nodes, node names: p1, p2, ...
  \path
    (0, 0) coordinate (M)
    \foreach \t [count=\i] in {1, ..., \MainNum} {
      +({1- \i)*360/\MainNum + \MainStartAngle}:\MainRadius)
      node[main node, align=center] (y_\i) {$ y_{\t}$}
    }
  ;  
    \path
    (0, 0) coordinate (M)
    \foreach \t [count=\i] in {1, ..., \MainNum} {
      +({1- \i)*360/\MainNum + \MainStartAngle}:\SecondRadius)
      node[main node, align=center] (bary_\i) {$\bar y_{\t}$}
    }
  ;  

  % Calculate the angle between the equal sides of the triangle
  % with side length \MainRadius, \MainRadius and radius of circle node
  % Result is stored in \p1-angle, \p2-angle, ...
  \foreach \i in {1, ..., \MainNum} {
    \pgfextracty{\dimen0 }{\pgfpointanchor{y_\i}{north}} 
    \pgfextracty{\dimen2 }{\pgfpointanchor{y_\i}{center}}
    \dimen0=\dimexpr\dimen2 - \dimen0\relax 
    \ifdim\dimen0<0pt \dimen0 = -\dimen0 \fi
    \pgfmathparse{2*asin(\the\dimen0/\MainRadius/2)}
    \global\expandafter\let\csname p\i-angle\endcsname\pgfmathresult
  }
  
    \foreach \i in {1, ..., \MainNum} {
    \pgfextracty{\dimen0 }{\pgfpointanchor{bary_\i}{north}} 
    \pgfextracty{\dimen2 }{\pgfpointanchor{bary_\i}{center}}
    \dimen0=\dimexpr\dimen2 - \dimen0\relax 
    \ifdim\dimen0<0pt \dimen0 = -\dimen0 \fi
    \pgfmathparse{2*asin(\the\dimen0/\SecondRadius/2)}
    \global\expandafter\let\csname q\i-angle\endcsname\pgfmathresult
  }

  % Draw the arrow arcs
  \foreach \i [evaluate=\i as \nexti using {int(mod(\i, \MainNum)+1}]
  in {1, ..., \MainNum} {  
    \pgfmathsetmacro\StartAngle{   
      (\i-1)*360/\MainNum + \MainStartAngle
      + \csname p\i-angle\endcsname
    }
    \pgfmathsetmacro\EndAngle{
      (\nexti-1)*360/\MainNum + \MainStartAngle
      - \csname p\nexti-angle\endcsname
    }
    \ifdim\EndAngle pt < \StartAngle pt
      \pgfmathsetmacro\EndAngle{\EndAngle + 360}
    \fi
    \draw[<-]
      (M) ++(\StartAngle:\MainRadius)
      arc[start angle=\StartAngle, end angle=\EndAngle, radius=\MainRadius]
    ;
  }
  
    \foreach \i [evaluate=\i as \nexti using {int(mod(\i, \MainNum)+1}]
  in {1, ..., \MainNum} {  
    \pgfmathsetmacro\StartAngle{   
      (\i-1)*360/\MainNum + \MainStartAngle
      + \csname q\i-angle\endcsname
    }
    \pgfmathsetmacro\EndAngle{
      (\nexti-1)*360/\MainNum + \MainStartAngle
      - \csname q\nexti-angle\endcsname
    }
    \ifdim\EndAngle pt < \StartAngle pt
      \pgfmathsetmacro\EndAngle{\EndAngle + 360}
    \fi
    \draw[<-]
      (M) ++(\StartAngle:\SecondRadius)
      arc[start angle=\StartAngle, end angle=\EndAngle, radius=\SecondRadius]
    ;
    \draw [->] (bary_\nexti)to [bend right=30]   (y_\i);
     \draw [->] (y_\nexti)to [bend right=30]   (bary_\i);
  }

%\node[above left=1.4cm and 10pt of species_2] (species_0) 
%  {$S_1$};
%\node[below left=1.4cm and 10pt of species_2] (species_5) 
%  {$S_6$};
%\node[right=2cm of species_4](equation_2) 
%    {$S_2 + S_3 \rightarrow S_4 $}; 
%\node[above=10pt of equation_2] (equation_1)
%    {$S_1 + S_5 \rightarrow S_2 $}; 
%\node[below=10pt of equation_2] (equation_3) 
%    {$S_4 + S_5 \rightarrow S_6 $}; 
%
%\draw[->]
%  (species_0) to[out=-20,in=180] (species_1);
%\draw[<-]
%  (species_5) to[out=20,in=180] (species_3);
\end{tikzpicture}
\caption{ The quiver $\Q_5$ corresponding to recutting of pentagons.}\label{Fig3}
\end{figure}
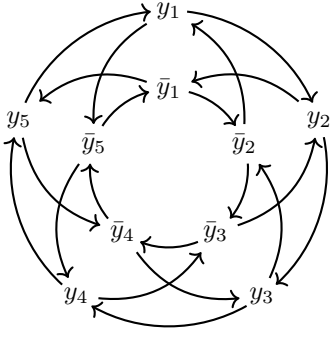

As a consequence of this cluster description, recutting on the space $\mathcal P_n^S  / S$ of planar polygons closed up to similarity has an invariant Poisson structure, namely the standard log-canonical structure defined by the quiver $\Q_n$. However, we are not aware of any invariant functions (first integrals) of recutting on $\mathcal P_n^S  / S$ besides the conjugacy class of the monodromy (whose preservation in particular means that the angle sum of the polygon is conserved) and, for even $n$, the sum of angles at every second vertex. To obtain additional invariants, we consider a smaller space $\mathcal P_n^E$ of planar polygons closed up to isometry (i.e. composition of rotations and translations). Our main result in that setting is the following:

  \begin{theorem}\label{thm1}
The recutting action of the group $\tilde S_n$ on the $2n$-dimensional space $\mathcal P_n^E  / E$ of planar $n$-gons closed up orientation-preserving isometry modulo said isometries is Arnold-Liouville integrable. Specifically, one has the following:
\begin{enumerate}
\item The recutting action on $\mathcal P_n^E  / E$  has an invariant Poisson structure and $\lfloor 3n/2\rfloor +~1$ independent first integrals (invariant functions). Out of those integrals, $2 \lfloor n/2 \rfloor + 2$ are Casimirs, so that the number of additional integrals is $\lceil n/2 \rceil - 1$, i.e. half of the dimension of symplectic leaves.
\item A generic joint level set of first integrals is a finite union of tori of dimension $\lceil n/2 \rceil - 1$. For each such torus $K$, the subgroup $G_{K} := \{ \omega \in \tilde S_n \mid \omega(K) \subset K\}$ of the recutting group elements preserving $K$  has a finite index in $\tilde S_n$. There is a flat structure on $K$ such that the action of~$G_{K}$ on $K$ is by translations. %In particular, for every $\rho \in \tilde S_n$ there is $k \in \Z$ such that~$\rho^k\vert_{K}$ is a translation.
\end{enumerate}
% it has as invariant Poisson structure and $\lfloor 3n/2\rfloor +~1$ independent invariants, out which $2 \lfloor n/2 \rfloor + 2$ are Casimirs. %, namely $n$ symmetric functions of the squared side lengths, the area, and $\lfloor n/2\rfloor - 2$ additional invariants. 

 \end{theorem}
 
%  The $2 \lfloor n/2 \rfloor + 2$ Casimir invariants are symmetric functions of squared side lengths, the angle sum, and, if $n$ is even, the sum of angles at every second vertex. The remaining $\lceil n/2 \rceil - 1$

 We prove  Theorem \ref{thm1} in Section \ref{sec:ial}. Note that the second part of the theorem (quasi-periodic dynamics on tori) is a standard consequence of the first one, so most of the section is devoted to the proof of the first part (existence of invariant Poisson structure, invariants, and their independence). The idea of the proof is based on the connection between recutting and refactorization of quaternionic polynomials. Note that due to non-commutativity of the skew field $\H$ of quaternions, a typical polynomial over $\H$ can be factored into linear factors in many different ways. In particular, a generic quadratic polynomial over $\H$ has two different factorizations. What we show is that the map interchanging those two factorizations can be geometrically interpreted as recutting. % This interpretation of recutting yields both an invariant Poisson structure and Poisson commuting first integrals. Namely, the integrals are defined by central functions of quaternionic polynomials (the absolute value and the real part), while the Poisson structure is %defined via an $r$-matrix of trigonometric type. %, so that both its invariance and Poisson commutativity of the integrals 
 As a result, integrability of recutting comes as a consequence of algebraic properties of quaternionic polynomials combined with some Poisson-Lie theory.
 
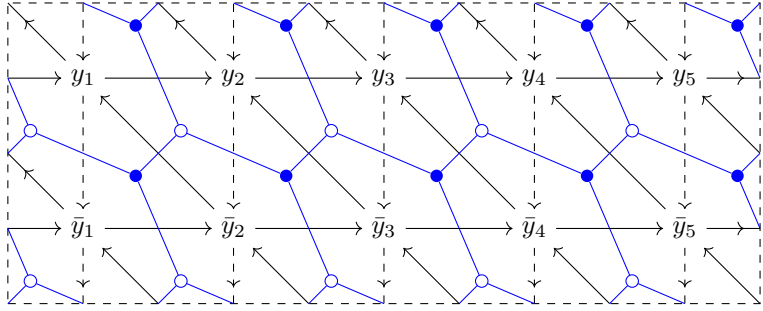
\begin{figure}
\centering
\begin{tikzpicture}
%\node () at (-1,-0.75) {$\Q_n$};
\node (A) at (0,0) {$y_{1}$};
\node (B) at (2,0) {$y_{2}$};
\node (C) at (4,0) {$y_{3}$};
\node (D) at (6,0) {$y_{4}$};
\node (E) at (8,0) {$y_{5}$};
\node (Ap) at (0,-2) {$\bar y_{1}$};
\node (Bp) at (2,-2) {$\bar y_{2}$};
\node (Cp) at (4,-2) {$\bar y_{3}$};
\node (Dp) at (6,-2) {$\bar y_{4}$};
\node (Ep) at (8,-2) {$\bar y_{5}$};
\draw [->, dashed] (A) -- (Ap);
\draw [->, dashed] (B) -- (Bp);
\draw [->, dashed] (C) -- (Cp);
\draw [->, dashed] (D) -- (Dp);
\draw [->, dashed] (E) -- (Ep);
\draw [->-, dashed] (Ap) -- +(0,-1);
\draw [->-, dashed] (Bp) -- +(0,-1);
\draw [->-, dashed] (Cp) -- +(0,-1);
\draw [->-, dashed] (Dp) -- +(0,-1);
\draw [->-, dashed] (Ep) -- +(0,-1);
\draw [<-, dashed] (A) -- +(0,1);
\draw [<-, dashed] (B) -- +(0,1);
\draw [<-, dashed] (C) -- +(0,1);
\draw [<-, dashed] (D) -- +(0,1);
\draw [<-, dashed] (E) -- +(0,1);
\draw [<-] (A) -- +(-1,0);
\draw [<-] (Ap) -- +(-1,0);
\draw [->] (A) -- (B);
\draw [->] (B) -- (C);
\draw [->] (C) -- (D);
\draw [->] (D) -- (E);
\draw [->-] (E) -- +(1,0);
\draw [->-] (Ep) -- +(1,0);
\draw [->] (Ap) -- (Bp);
\draw [->] (Bp) -- (Cp);
\draw [->] (Cp) -- (Dp);
\draw [->] (Dp) -- (Ep);
\draw [->-] (Ap) --+(-1,1);
\draw [->] (Bp) -- (A);
\draw [->] (Cp) -- (B);
\draw [->] (Dp) -- (C);
\draw [->] (Ep) -- (D);
\draw [<-] (E) -- +(1,-1);
\draw [->-] (A) --+(-1,1);
\draw [->-] (B) --+(-1,1);
\draw [->-] (C) --+(-1,1);
\draw [->-] (D) --+(-1,1);
\draw [->-] (E) --+(-1,1);
\draw [<-] (Ap) --+(1,-1);
\draw [<-] (Bp) --+(1,-1);
\draw [<-] (Cp) --+(1,-1);
\draw [<-] (Dp) --+(1,-1);
\draw [<-] (Ep) --+(1,-1);
\draw [dashed] (-1,1) -- (9,1) -- (9,-3) -- (-1,-3) -- cycle;

\node [circle, fill = blue, scale=0.5] (B1) at (0.7,-1.3) {};
\node [circle, fill = blue, scale=0.5] (B2) at (2.7,-1.3) {};
\node [circle, fill = blue, scale=0.5] (B3) at (4.7,-1.3) {};
\node [circle, fill = blue, scale=0.5] (B4) at (6.7,-1.3) {};
\node [circle, fill = blue, scale=0.5] (B5) at (8.7,-1.3) {};
\node [circle, fill = blue, scale=0.5] (B6) at (0.7,0.7) {};
\node [circle, fill = blue, scale=0.5] (B7) at (2.7,0.7) {};
\node [circle, fill = blue, scale=0.5] (B8) at (4.7,0.7) {};
\node [circle, fill = blue, scale=0.5] (B9) at (6.7,0.7) {};
\node [circle, fill = blue, scale=0.5] (B10) at (8.7,0.7) {};
\node [circle, draw = blue, scale=0.5] (W0) at (-0.7,-0.7) {};
\node [circle, draw = blue, scale=0.5] (W1) at (1.3,-0.7) {};
\node [circle, draw = blue, scale=0.5] (W2) at (3.3,-0.7) {};
\node [circle, draw = blue, scale=0.5] (W3) at (5.3,-0.7) {};
\node [circle, draw = blue, scale=0.5] (W4) at (7.3,-0.7) {};
\node [circle, draw = blue, scale=0.5] (W5) at (-0.7,-2.7) {};
\node [circle, draw = blue, scale=0.5] (W6) at (1.3,-2.7) {};
\node [circle, draw = blue, scale=0.5] (W7) at (3.3,-2.7) {};
\node [circle, draw = blue, scale=0.5] (W8) at (5.3,-2.7) {};
\node [circle, draw = blue, scale=0.5] (W9) at (7.3,-2.7) {};
\draw[blue] (W0) -- (B1) -- (W1) -- (B2) -- (W2) -- (B3) -- (W3) -- (B4) -- (W4) -- (B5);
\draw[blue] (B6) -- (W1);
\draw[blue] (B7) -- (W2);
\draw[blue] (B8) -- (W3);
\draw[blue] (B9) -- (W4);
\draw[blue] (B1) -- (W6);
\draw[blue] (B2) -- (W7);
\draw[blue] (B3) -- (W8);
\draw[blue] (B4) -- (W9);
\draw[blue] (B6) -- +(0.3,0.3);
\draw[blue] (B7) -- +(0.3,0.3);
\draw[blue] (B8) -- +(0.3,0.3);
\draw[blue] (B9) -- +(0.3,0.3);
\draw[blue] (B10) -- +(0.3,0.3);
\draw[blue] (B6) -- +(-0.7,0.3);
\draw[blue] (B7) -- +(-0.7,0.3);
\draw[blue] (B8) -- +(-0.7,0.3);
\draw[blue] (B9) -- +(-0.7,0.3);
\draw[blue] (B10) -- +(-0.7,0.3);
\draw[blue] (W5) -- +(-0.3,-0.3);
\draw[blue] (W6) -- +(-0.3,-0.3);
\draw[blue] (W7) -- +(-0.3,-0.3);
\draw[blue] (W8) -- +(-0.3,-0.3);
\draw[blue] (W9) -- +(-0.3,-0.3);
\draw[blue] (W5) -- +(+0.7,-0.3);
\draw[blue] (W5) -- +(-0.3,+0.7);
\draw[blue] (W6) -- +(+0.7,-0.3);
\draw[blue] (W7) -- +(+0.7,-0.3);
\draw[blue] (W8) -- +(+0.7,-0.3);
\draw[blue] (W9) -- +(+0.7,-0.3);
\draw[blue] (W0) -- +(-0.3,-0.3);
\draw[blue] (W0) -- +(-0.3,+0.7);
\draw[blue] (B5) -- +(0.3,0.3);
\draw[blue] (B5) -- +(0.3,-0.7);
\draw[blue] (B10) -- +(0.3,-0.7);
%\draw[blue] (B10) -- (W1);
%\node (D) at (0,-2) {$\bar y_{1}$};
%\node (E) at (2,-2) {$\bar y_2$};
%\node (F) at (4,-2) {$\bar y_{3}$};
%\draw [->] (A) -- (B);
%\draw [<-] (C) -- (B);
%\draw [<-] (E) -- (D);
%\draw [->] (E) -- (F);
%\draw [->] (B) -- (D);
%\draw [->] (C) -- (E);
%\draw [->] (E) -- (A);
%\draw [->] (F) -- (B);
%\draw [<-] (A.west) -- +(-0.5,0);
%\draw [->] (C.east) -- +(0.5,0);
%%\draw [->] (A) -- +(-0.66,-0.5);
%\draw [<-] (D.west) -- +(-0.5,0);
%\draw [->] (F.east) -- +(0.5,0);
%\draw [->] (D) -- +(-0.66,0.5);
\end{tikzpicture}
\caption{ Embedding of the quiver $\Q_5$ into a torus and its dual graph.}\label{FigT}
\end{figure}
 \begin{remark}
 The Poisson structure on polygons closed up to isometry (coming from quaternionic polynomials) is compatible with the cluster structure on polygons closed up to similarity in the sense that the natural map $\mathcal P_n^E  / E \to \mathcal P_n^S  / S $ is Poisson. Furthermore, one can relate those structures as follows. As can be seen in Figure~\ref{FigT}, the quiver $\Q_n$ embeds in a torus (the quiver is represented by solid black arrows, and the opposite sides of the dashed rectangle are identified). Furthermore, by enhancing the quiver $\Q_n$ with obsolete arrows from each $y_i$ to the corresponding $\bar y_i$ and back (dashed black arrows in the figure), one gets a quiver with bipartite dual $\Q_n^*$ (the blue graph in the figure). As a result, one can use the techniques of \cite{GSTV, GK} to build a Poisson structure on the space $\Sigma(\Q_n^*)$ of edge weights of $\Q_n^*$ modulo gauge transformations. Furthermore, there is a natural way to identify the space $\mathcal P_n^E  / E $ with a Poisson hypersurface in $\Sigma(\Q_n^*)$, so that the Poisson property of the map $\mathcal P_n^E  / E \to \mathcal P_n^S  / S $ becomes a consequence of the Poisson property of the map from edge weights to face weights.

Interpreted in terms of the dual graph  $\Q_n^*$, recutting becomes a certain non-local transformation of a weighted bipartite graph on a torus known as the \textit{plabic $R$-matrix} \cite{chepuri2020plabic}. Invariants of such a transformation can be constructed by using either the dimer partition function \cite{GK}, or the boundary measurement matrix \cite{GSTV} (as shown in \cite{GAFA}, those two approaches give the same invariants). This gives an alternative route to proving Theorem \ref{thm1}. We choose  not to pursue this approach since our construction based on quaternionic polynomials seems more direct and also adapts better to the case of closed polygons that we consider next.

%As shown in \cite{I}
 \end{remark}
 
 Our Poisson structure on polygons closed up to isometry restricts to polygons closed up to translation, so one can prove integrability of recutting in the latter setting following the lines of the proof of Theorem \ref{thm1} (note that Theorem \ref{thm1} does not directly imply integrability for any smaller class of polygons since it only describes the behavior of recutting on \textit{generic} level sets of invariants). However, this approach does not work for closed polygons, since such polygons {do not} constitute a Poisson submanifold. One way to overcome this difficulty is by using Dirac reduction. Here we take a different approach, namely we show that although one cannot restrict the Poisson structure to closed polygons, one can still restrict the Hamiltonian vector fields generated by the invariants, which is sufficient to establish integrability in the non-Hamiltonian setting. Our main result for closed polygons is the following:

 \begin{theorem}\label{thm2}
 
 Assume that $n \geq 3$. Then the recutting action of the group $\tilde S_n$ on the $2n-3$-dimensional space $\mathcal P_n  / E$ of closed planar $n$-gons modulo orientation-preserving isometries is integrable in the non-Hamiltonian sense. Specifically, one has the following:
\begin{enumerate}
\item The recutting action on $\mathcal P_n  / E$  has $\lfloor 3n/2\rfloor -~1$ independent first integrals and a complementary number $\lceil n/2\rceil -~2$ of independent invariant commuting vector fields tangent to  level sets of first integrals.

\item A generic joint level set of first integrals is a finite union of tori of dimension $\lceil n/2 \rceil - 2$. For each such torus $K$, the subgroup $G_{K} := \{ \omega \in \tilde S_n \mid \omega(K) \subset K\}$ of the recutting group elements preserving $K$ has a finite index in $\tilde S_n$. There is a flat structure on $K$ such that the action of~$G_{K}$ on $K$ is given by translations. 
 \end{enumerate}

% Assume that $n \geq 3$. Then the action of the recutting group $G_R$ on the space of closed $n$-gons considered up to orientation-preserving isometries is completely integrable. Specifically, it has $\lfloor 3n/2\rfloor -~1$ algebraically independent invariants %, namely $n$ symmetric functions of the squared side lengths, the area, and $\lfloor n/2\rfloor - 2$ additional invariants. 
%whose generic joint level set is a finite union of tori of dimension $\lceil n/2 \rceil - 2$. For each such torus $T$, the subgroup $G_R^T := \{ \rho \in G_R \mid \rho(T) \subset T\}$ of the recutting group elements preserving $T$ is a commutative subgroup of finite index. There is a flat structure on $T$ such that the action of~$G_R^T$ on $T$ is given by translations. In particular, for every $\rho \in G_R$ there is $k \in \Z$ such that~$\rho^k\vert_{T}$ is a translation.
 \end{theorem}
 %This theorem implies that for $n=3,4$ the orbit of the recutting group action on triangles and quadrilaterals is a finite set. That is obviously for a triangle, as the only thing that recutting does in that case is 
  
  %In both cases, the recutting action of $\tilde S_n$ factors through the map  $\tilde S_n \to S_n$ assigning to every recutting the corresponding permutation of lengths of sides. For a triangle, the orbit consists of $3! = 6$ triangles, $3$ of which are obtained from the initial one by cyclic relabeling of vertices, while $3$ remaining ones are their reflected copies. Similarly, the orbit of a quadrilateral contains $4! = 24$ triangles which fall into three geometrically distinct classes, see Remark \ref{rem:quad}.
 For example, for triangles and quadrilaterals we get $\lceil n/2 \rceil - 2 = 0$, so the orbits consist of finitely many points, cf. Remark \ref{rem:quad} below. For $n \geq 5$ the orbits are likely to be infinite.
 
Out of $\lfloor 3n/2\rfloor -~1$ invariants, $n + 1$ have a clear geometric meaning: they are symmetric functions of the squared lengths of sides, and the area of the polygon. In addition, when $n$ is even, one of the invariants is the sum of angles at every second vertex. The remaining invariants do not seem to have such a transparent interpretation.\par
 We prove Theorem \ref{thm2} in Section \ref{sec:closed}.

\medskip

\medskip
{\bf Acknowledgments.} The author is grateful to Vsevolod Adler, Maxim Arnold, Michael Gekhtman, Boris Khesin, Pavlo Pylyavskyy, Sanjay Ramassamy, Richard Schwartz, Alexander Shapiro, Sergei Tabachnikov, Alexander Veselov, and the anonymous referee for fruitful discussions and useful remarks. This work was supported by NSF grant DMS-2008021.

%It is, however, a priori unclear why these functions are preserved by the sequence of mutations corresponding to recutting. 
%\section{Preliminaries}

 %$Y_{\Q}^\R   \to  \dots \to Y_{\tilde \Q}^\R$ and a map $Y_{\tilde \Q}^\R \to Y_{\Q}^\R$ induced by a real isomorphism $\tilde \Q \to \Q$.  
\section{Polygon spaces and recutting}\label{sec:ps}
For the purposes of the present paper, a \textit{polygon} is a bi-infinite sequence $(v_i \in \C)_{i \in \Z}$ such that $v_i \neq v_{i+1}$ for all $i \in \Z$. \textit{Recutting} $\rho_j$ of a polygon  $(v_i)$ at a vertex $v_j$ is defined as reflection of $v_j$  in the perpendicular bisector of the interval $(v_{j-1}, v_{j+1})$. It is well-defined as long as $v_{j-1} \neq v_{j+1}$.\par
The space $\mathcal P$ of polygons carries the action of the group $S := \{\psi \colon \C \to \C \mid \psi(z) = \alpha z + \beta,  \alpha \in \C^*, \beta \in \C \}$ of orientation-preserving similarities, as well as of its normal subgroups $E := \{\psi \colon \C \to \C \mid \psi(z) = \alpha z + \beta,  \alpha \in S^1, \beta \in \C \}$ of orientation-preserving isometries and $T :=  \{\psi \colon \C \to \C \mid \psi(z) =  z + \beta,  \beta \in \C \}$ of translations. All these actions commute with recutting. \begin{remark}Throughout the paper, all similarities and isometries are assumed to be orientation-preserving, so we refer to elements of $S$ and $E$ as simply similarities and isometries. \end{remark}
For a polygon $(p_i) \in \mathcal P$, define its \textit{edge vectors} $z_i \in \C$ by $z_i := v_i - v_{i-1}$. Let also $y_i := z_{i+1}/z_i$ be the ratios of consecutive edge vectors, and let $\phi_i := \mathrm{arg}(y_i)$ be the angles between consecutive edge vectors (for a convex counter-clockwise oriented polygon, $\phi_i$ can be understood as exterior angles, so will we often refer to them in that way). Then $y_i$ parametrize polygons modulo similarities (i.e. the action of $S$), edge vectors $z_i$ parametrize polygons modulo translations (i.e. the action of $T$), while $\phi_i$ and~$|z_i|$ parametrize polygons modulo  isometries (i.e. the action of $E$). As another parametrization of polygons modulo isometries we will use the sequence of edge vectors $z_i$ modulo simultaneous rotations.

The following results express recutting in terms of coordinates $z_i$ and $y_i$.

 \begin{lemma}\label{lemma:zrels}

Assume that a polygon $(v_i')$  is the image of a polygon $(v_i)$ under recutting $\rho_j$ at $v_j$. Let  $z_i = v_i - v_{i-1}$ be the edge vectors of $(v_i)$ and  $z_i' := v_i' - v_{i-1}'$ be the edge vectors of $(v_i')$. Then
\begin{equation}\label{eq:zrels}
\begin{gathered}
z_j' + z_{j+1}' = z_j+z_{j+1},\\
z_j'\bar z_{j+1}' = z_j\bar z_{j+1},
\end{gathered}
%z_j' = \bar z_{j+1}\frac{z_j + z_{j+1}}{\bar z_j + \bar z_{j+1}},\\
%z_{j+1}' = \bar z_j\frac{z_j + z_{j+1}}{\bar z_j + \bar z_{j+1}}.
\end{equation}
where $\bar z$ stands for the complex conjugate of $z$.
%
%\begin{gather*}
%z_j' = \bar z_{j+1}\frac{z_j + z_{j+1}}{\bar z_j + \bar z_{j+1}},\\
%z_{j+1}' = \bar z_j\frac{z_j + z_{j+1}}{\bar z_j + \bar z_{j+1}}.
%\end{gather*}
\end{lemma}
\begin{proof}
The complex number $z_j+z_{j+1}$ represents the side $(v_{j-1}, v_{j+1})$ of the triangle $(v_{j-1},v_j, v_{j+1})$. As for $z_j\bar z_{j+1}$, its absolute value is the product of lengths of  $(v_{j-1}, v_{j})$ and  $(v_{j}, v_{j+1})$, while its argument is the exterior angle of the triangle $(v_{j-1},v_j, v_{j+1})$ at $v_j$. None of these change when the triangle is cut and then reattached with opposite orientation, hence the result. 
\end{proof}
\begin{corollary}\label{cor:first}
As a transformation of the space of polygons modulo translations, recutting $\rho_j$ is given by
\begin{equation}\label{eq:zrels2}
\begin{gathered}
z_j' = \bar z_{j+1}\frac{z_j + z_{j+1}}{\bar z_j + \bar z_{j+1}},\\
z_{j+1}' = \bar z_j\frac{z_j + z_{j+1}}{\bar z_j + \bar z_{j+1}},
\end{gathered}
\end{equation}
and $z_i' = z_i$  for $i \neq j, j+1$.
\end{corollary}
\begin{proof}
Using Lemma \ref{lemma:zrels} along with the relation $|z_j'| = |z_{j+1}|$, one gets
\begin{equation*}
z_j + z_{j+1} = z_j' + z_{j+1}' = z_j' + \frac{\bar z_j z_{j+1}}{\bar z_j'} =  \frac{z_j'\bar z_j' + \bar z_j z_{j+1}}{\bar z_j'} =  \frac{z_{j+1}\bar z_{j+1}+ \bar z_j z_{j+1}}{\bar z_j'} ,
\end{equation*}
which implies the desired formula for $z_j'$. The formula for $z_{j+1}'$ now follows from any of the relations \eqref{eq:zrels}. Other~$z_i$ do not change under recutting $\rho_j$ since they do not depend on the vertex $v_j$.
\end{proof}

 \begin{corollary}\label{cor:second}

As a transformation of the space of polygons modulo similarities, recutting $\rho_j$ is given by
\begin{equation}\label{eq:yrels}
\begin{gathered}
 y_{j-1}' = \frac{y_{j-1}(1+y_j)}{1 + \bar y_j^{-1}},\\
y_j' = \bar y_j^{-1},\\
y_{j+1}' = \frac{y_{j+1}(1+\bar y_j)}{1 +  y_j^{-1}},
\end{gathered}
%z_j' = \bar z_{j+1}\frac{z_j + z_{j+1}}{\bar z_j + \bar z_{j+1}},\\
%z_{j+1}' = \bar z_j\frac{z_j + z_{j+1}}{\bar z_j + \bar z_{j+1}}.
\end{equation}
and $y_i' = y_i$ for $i \neq j-1, j, j+1$.
 \end{corollary}
 \begin{proof}
 Straightforward calculation using \eqref{eq:zrels2} along with the definitions $y_i = z_{i+1}/z_i$ and  $y_i' := z_{i+1}'/z_i'$ of $y$ coordinates.
 \end{proof}
 
 In the rest of the paper, we only consider polygons satisfying certain periodicity-type conditions. Namely, let $\psi \in S$ be a similarity transformation. A polygon $(v_i)$ is said to be an \textit{$n$-gon closed up to $\psi$} (or an $n$-gon with \textit{monodromy}~$\psi$) if $p_{i+n} = \psi(v_i)$ for all $i \in \Z$. 
 Recutting $\rho_j$ of an $n$-gon $(v_i)$ with monodromy $\psi$ is a polygon with the same monodromy obtained from $(v_i)$ by recutting at all vertices $v_i$ with $i \equiv j \mod n$. 
 
  For a subgroup $G \subset  S$, denote by $\mathcal P_n^{G}$ the space of $n$-gons with monodromy $\psi \in G$ (we also use the notation $\mathcal P_n$ for the space of closed $n$-gons corresponding to the trivial group $G$). Note that since the recutting action on $\mathcal P_n^{G}$ preserves the monodromy and commutes with similarity transformations, it descends to a self-map of the quotient $ \mathcal P_n^{G} / H$ where $H \subset  S$ is any subgroup normalizing $G$. In what follows, we will be particularly interested in the action of recutting on the spaces $ \mathcal P_n^{S} / S$ and $ \mathcal P_n^{E} / E$. The following results are straightforward: \begin{proposition}\label{prop:ycoords}
 The assignment $(v_i) \mapsto (y_i)$ taking a polygon to the ratios of its consecutive edge vectors gives a bijection
 $$
\mbox{\raisebox{-.3em}{$\faktor{ \mathcal P_n^{S}}{ S}$}}\simeq \{\mbox{$n$-periodic sequences $y_i \in \C^*$}\}.
 $$
  Written in coordinates $y_i$, recutting $\rho_j $  on $P_n^{S}/ S$ is given by \eqref{eq:yrels}. % where all indices are understood modulo $n$. 
 \end{proposition}
% \begin{proof}
% Given a polygon $(p_i) \in \mathcal P_n^{S}$, we have $p_{i+n} = \alpha p_i + \beta$, so $z_{i+n} = \alpha z_i$ and $y_{i+n} = y_i$. Thus, the ratios of consecutive edge vectors provide a map
% \end{proof}
 
 %Clearly, the space $ \mathcal P_n^{E}/E$ and the action of recutting on that space can be described in a similar way:
 %Analogously, we have the following:
 \begin{proposition}\label{prop:zphicoords}
 The assignment $(v_i) \mapsto (|z_i|, \phi_i)$ taking a polygon to its side lengths and exterior angles gives a bijection
 $$
\mbox{\raisebox{-.3em}{$\faktor{ \mathcal P_n^{E}}{ E}$}}\simeq \{\mbox{pairs of $n$-periodic sequences $|z_i| \in \R_+$, $\phi_i \in \R / 2\pi \Z$}\},
 $$
 while the assignment $(v_i) \mapsto (z_i)$ taking a polygon to its edge vectors gives a bijection
  $$
\mbox{\raisebox{-.3em}{$\faktor{ \mathcal P_n^{E}}{ E}$}}\simeq \faktor{\{\mbox{sequences $z_i \in \C^* \mid z_{i+n} = \alpha z_i$ for some $\alpha \in S^1$}\}}{S^1},
 $$
 where $S^1$ stands for the set of complex numbers of absolute value $1$. Written in terms of $z_i$, recutting $\rho_j $ on $  \mathcal P_n^{E} / E $ is given by \eqref{eq:zrels2}.% where all indices are understood modulo $n$. 
 \end{proposition}
%Similarly, as coordinates on  $ \mathcal P_n^{E}/E$, one can use either the entries of periodic sequences $|z_i|$, $\phi_i$, or the entries of the sequence $z_i$, considered up to simultaneous rotation. In the latter case, the sequence $z_i$ is not periodic but satisfies the quasi-periodicity condition $z_{i+n} = \alpha z_i$ with $\alpha \in S^1$. Coordinate expression for recutting $\rho_j$ on  $ \mathcal P_n^{E} / SE$ is given by \eqref{eq:zrels2} where, again, all indices are understood modulo $n$.

%
% and commutes with the action of the group  $S$ of orientation-preserving similarities. Therefore, $\rho_j$ can be viewed a densely defined self-map of the quotient  ${ \mathcal P_n^{S}}/{ S}$ (it is densely defined since the definition of recutting requires $v_{i-1} \neq v_{i+1}$ for $i \equiv j \mod n$), and hence can be expressed in coordinates $y_i$. Explicitly, recutting $\rho_j$ is given by formulas \eqref{eq:yrels} where indices are now understood modulo $n$.
 \section{Recutting of polygons closed up to similarity: cluster structure}\label{sec:cluster}
 
\subsection{Quivers, mutations, and real structures}\label{sec:qmrs}
In this section we discuss the general notion of a $Y$-type (also known as $\mathcal X$-type) cluster mutation, with an emphasis on quivers endowed with an involution (a \textit{real structure}). Recall that a \textit{quiver} is a directed graph without loops or oriented cycles of length $2$. For simplicity, in what follows we also prohibit multiple edges. Given a quiver $\Q$ with the vertex set $\{1, \dots, n\}$, denote by $Y_{\Q}$ the space of functions $\{1, \dots, n\} \to \C^*$. The space $Y_{\Q}$ is a complex torus of dimension $n$. It comes equipped with canonical coordinates $y_1, \dots, y_{n}$ given by evaluation of functions at vertices of $\Q$: for $\xi \in Y_{\Q}$, one defines $y_i(\xi) := \xi(i)$. Since the variables $y_i$ are indexed by vertices of $\Q$, in what follows we often identify vertices with the corresponding $y$ variables.\par

The torus $Y_{\Q}$ carries a Poisson structure. In terms of $y_i$ coordinates, it has a \textit{log-canonical} form
\begin{equation}\label{lcpb}
\{y_i, y_j\} = a_{ij}y_iy_j,
\end{equation}
where $(a_{ij})$ is the signed adjacency matrix of $\Q$, i.e. 
\begin{equation*}
a_{ij} = \left[\begin{aligned}
1, \quad &\mbox{if there is an arrow (a directed edge) from vertex $i$ to vertex $j$}, \\
-1, \quad &\mbox{if there is an arrow vertex $j$ to vertex $i$},\\
0, \quad &\mbox{if the vertices $i$, $j$ are not connected by an arrow}.
\end{aligned} \right.
\end{equation*}
 The Poisson structure on $Y_{\Q}$ is natural in the following sense: any isomorphism of quivers $\Q \to \Q'$ induces a Poisson isomorphism $Y_{\Q } \to Y_{\Q'} $.\par
Given a quiver $\Q$, the \textit{quiver mutation} of $\Q$ at its $i$'th vertex is the following modification of $\Q$:%an operation producing a new quiver $\Q'$ with the same vertex set as  $\Q$, constructed as follows:
\begin{enumerate}
\item For every pair of vertices $j,k $ of $\Q$ such that there is an arrow from $j$ to $i$ and an arrow from $i$ to $k$, add an arrow from $j$ to $k$.
\item Reverse all arrows adjacent to the vertex $i$.
\item Remove all newly formed oriented cycles of length $2$.
\end{enumerate}
The result of a quiver mutation of $\Q$ is a new quiver $\Q'$ with the same vertex set as $\Q$. Note that in general a quiver mutation produces a quiver with multiple edges. This, however, does not happen for quivers relevant to the present paper.\par
We now define the notion of a $Y$-mutation. Assume that a quiver $\Q'$ is obtained from  $\Q$ by means of mutation at vertex $i$. The corresponding \textit{$Y$-mutation} is a birational map $Y_{\Q } \to Y_{\Q'} $ defined as follows. Let $y_j$ be the canonical coordinates in $Y_{\Q }$, and $y_j'$ be the canonical coordinates in $Y_{\Q'}$. Expressed in these coordinates, the $Y$-mutation $Y_{\Q } \to Y_{\Q'} $ at $i$ is given by
\begin{equation*}
y_j' = \left[\begin{aligned} &y_j^{-1}, \quad \mbox{if $j = i$},\\
&y_j(1 + y_i^{-1})^{-1}, \quad \mbox{if there is an arrow from vertex $i$ to vertex $j$},\\
&y_j(1 + y_i), \quad \mbox{if there is an arrow from vertex $j$ to vertex $i$},\\
&y_j,\quad \mbox{in all other cases}.
\end{aligned} \right.
\end{equation*}
A $Y$-mutation $Y_{\Q } \to Y_{\Q'} $ is a Poisson map.\par
\begin{figure}[t]
\centering
\begin{tikzpicture}[]
%\node [draw,circle,color=black, fill=white,inner sep=0pt,minimum size=5pt] (A) at (0,-0.5) {};
%\draw [<-]   (A) -- +(-1,0) coordinate (SRC);
%\node [draw,rectangle,color=black, fill=white,inner sep=0pt,minimum size=5pt] (SRCN) at (SRC) {$$};
%\path  (A) -- +(2,2) coordinate (Bp);
%\node [draw,circle,color=black, fill=black,inner sep=0pt,minimum size=5pt] (B) at (Bp) {};
%\path [->]  (B) -- +(1,0) coordinate (SINK);
%\node [draw,rectangle,color=black, fill=white,inner sep=0pt,minimum size=5pt] (SINKN) at (SINK) {};
%\draw [->]  (B) -- (SINKN);
%\draw [->, decorate] (A) -- (B);
\node () at (0,0)
{
\begin{tikzpicture}[]
\node  (A) at (0,0) {$y_1$};
\node (B) at (1,0) {$y_2$};
\node (C) at (2,0) {$y_3$};
\node(D) at (3,0) {$y_4$};
\draw [->] (A) -- (B);
\draw [->] (B) -- (C);
\draw [->] (C) -- (D);
\draw [->, dashed] (4.2,0) -- (5.2,0);
\end{tikzpicture}
};
\node () at (7,-0.2)
{
\begin{tikzpicture}[]
\node (A) at (0,0) {$ {y_1}(1+y_2)$};
\node (B) at (2,0) { $y_2^{-1}$};
\node (C) at (4,0) {$\displaystyle \frac {y_3}{1 + y_2^{-1}}$};
\node (D) at (6,0) {$y_4$};
\draw [->] (B) -- (A);
\draw [->] (C) -- (B);
\draw [->] (C) -- (D);
\draw [->] (A) .. controls (1,-0.7) and (3,-0.7) .. (C);
\end{tikzpicture}
};
%\draw [  line width = 3pt] (A) -- (D);
%\draw [ line width = 3pt] (B) -- (C);
%\draw [<-]   (A) -- +(-1,0) coordinate (SRC);
%\node [draw,circle,color=black, fill=black,inner sep=0pt,minimum size=5pt, label = source] (SRCN) at (SRC) {};
%\path  (A) -- +(2,2) coordinate (Bp);
%\node [draw,circle,color=black, fill=black,inner sep=0pt,minimum size=5pt] (B) at (Bp) {};
%\path [->]  (B) -- +(1,0) coordinate (SINK);
%\node [draw,circle,color=black, fill=white,inner sep=0pt,minimum size=5pt, label = sink] (SINKN) at (SINK) {};
%\draw [->]  (B) -- (SINKN);
%\draw [->, decorate] (A) -- (B);
%\path [->]  (SINKN) -- +(0,-3) coordinate (C);
%\node [draw,circle,color=black, fill=black,inner sep=0pt,minimum size=5pt] (CN) at (C) {};
%\draw [dashed, ->]  (SINKN) -- (CN);
%\path [->]  (SRCN) -- +(0,-1) coordinate (D);
%\node [draw,circle,color=black, fill=white,inner sep=0pt,minimum size=5pt] (DN) at (D) {};
%\draw [dashed, ->]  (DN) -- (SRCN);
%\draw [dashed, ->] (CN) -- (DN);
\end{tikzpicture}
\caption{A $Y$-mutation at $y_2$.}\label{Fig5}
\end{figure}
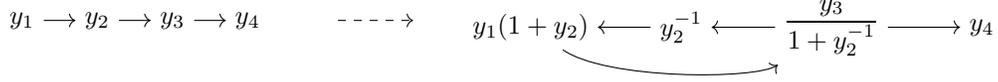
We depict $Y$-mutations as shown in Figure \ref{Fig5}. The labels at vertices of the initial quiver $\Q$ are the corresponding $Y$-variables $y_i$ while the labels at vertices of the mutated quiver $\Q'$ are pull-backs of the corresponding $Y$-variables $y_i'$ by the $Y$-mutation map $Y_{\Q } \to Y_{\Q'} $.

Now assume we have sequence of quivers $\Q  \to \dots \to \tilde \Q$ where each quiver is obtained from the previous one by mutation. Suppose also that we have an isomorphism $\psi \colon \tilde \Q \to \Q$. Then the composition $Y_{\Q} \to \dots \to Y_{\tilde \Q}$ of $Y$-mutations, followed by the map $Y_{\tilde \Q} \to Y_{\Q}$ induced by the isomorphism $\psi$, is a birational Poisson map of $Y_{\Q}$ onto itself. We call such a map a ($Y$-type) \textit{cluster transformation}. Put differently, a cluster transformation is such a sequence of mutations which, after permutation of vertices, restores the initial quiver. %a composition of $Y$-mutations followed by a permutation vertices. 
\par
We now add a real structure to the picture. Let $\tau \colon \Q \to \Q$ be an involution (i.e. a graph automorphism such that $\tau^2 = \mathrm{id}$).  Then $\tau$ defines a \textit{real structure} (i.e. an anti-holomorphic involution) $ \hat \tau$ on $Y_{\Q}$ by the rule  $\hat \tau(\xi) := \overline{\tau^*\xi}.$
In terms of coordinates $y_i$, the involution $\hat \tau$ is given by $\hat \tau^*y_i = \bar y_{\tau(i)}$. 
The \textit{real part} $ Y_{\Q}^\R$ of $ Y_{\Q}$   is the fixed point set of the involution $\hat \tau$. It is a real manifold whose complexification is  the complex torus $Y_{\Q}$ (in particular, $\dim_\R Y_{\Q}^\R = \dim_\C Y_{\Q}$ is the number of vertices of $\Q$). A function $\xi \in Y_{\Q}$ belongs to $ Y_{\Q}^\R$ if and only if it takes real values at vertices fixed by $\tau$ and complex conjugate values at vertices switched by $\tau$. %The real dimension of $ Y_{\Q}^\R$
The manifold  $ Y_{\Q}^\R$ is parametrized by $y_i$'s subject to relations $\bar y_i=  y_{\tau(i)}$ (in particular, $y_i$ is real if the vertex $i$ is fixed by $\tau$). \par
\begin{proposition}\label{prop:cr}
The Poisson structure on $ Y_{\Q}$ restricts to its real part $ Y_{\Q}^\R$. 
\end{proposition}
The proof is based on the following general lemma, which is also used later in the paper.
\begin{lemma}\label{lemma:cr}
Let $V$ be a real vector space endowed with a polynomial Poisson structure, and let $\sigma \colon V \to V$ be a linear Poisson involution. Define an anti-linear involution $\bar \sigma$ on $V_\C := V \otimes \C$ by  $\bar \sigma(x) :=  \sigma(\bar x)$, where $\sigma$ is extended from $V$ to $V_\C$ by $\C$-linearity. Let $V_\sigma := \mathrm{Fix}(\bar \sigma)$ be the fixed point set of $\bar \sigma$. Then there is a unique Poisson structure on the real vector space $V_\sigma$ whose complexification coincides with the complexification of the Poisson structure on $V$.
\end{lemma}
\begin{proof}[Proof of Lemma \ref{lemma:cr}]
The space $V_\sigma$ is a real form of the complex vector space $V_\C$, so there is at most one Poisson structure on $V_\sigma$ extending to the Poisson structure on $V_\C$. To prove existence, notice that we have an isomorphism of $\R$-algebras   $\R[V_\sigma] \simeq  \mathrm{Fix}(\bar \sigma^*)$, where  $ \bar \sigma^*\colon \C[V_\C] \to \C[V_\C]$ is an involution given by $(\bar \sigma^*f)(x) = \overline{f( \bar \sigma( x))}.$
So, to obtain the desired Poisson bracket on $\R[V_\sigma]$, it suffices to show that the $\R$-subalgebra $\mathrm{Fix}(\bar \sigma^*)$ is closed under the Poisson bracket on $\C[V_\C]$. To that end, observe that $ \bar \sigma^*$ is a composition of two commuting Poisson involutions: $f(x) \mapsto {f( \sigma(x))}$ and $ f(x) \mapsto \overline{f(\bar x)}$. So, $ \bar \sigma^*$ is itself a Poisson involution and its fixed point set $\mathrm{Fix}(\bar \sigma^*) = \R[V_\sigma]$ is indeed closed under the Poisson bracket, as desired.
%Let $\{\,,\}_\C$ be the Poisson structure on $V_\C$ obtained by extending  the Poisson structure $\{\,,\}$ by $\C$-linearity. Then $(\C[V_\C],\{\,,\}_\C) $ is a complex Poisson algebra. That algebra is endowed with two commuting involutions: a linear one defined by $(\tau^*F)(x) := {F( \tau(x))}$, and an anti-linear one defined by $\bar F (x) := \overline{F(\bar x)}$. Their composition is the involution $ \bar \tau^*$ given by $(\bar \tau^*F)(x) = \overline{F( \tau(\bar x))}.$ So, $\bar \tau^*$ is also a anti-linear involution of the complex Poisson algebra  $(\C[V_\C],\{\,,\}_\C) $, which implies that its fixed point set is a real form of $(\C[V_\C],\{\,,\}_\C) $. But fixed points of $\bar \tau^*$ are precisely extensions of polynomials $F \in \R[V_\tau]$, so $\R[V_\tau]$ is a real Poisson algebra whose complexification is  $(\C[V_\C],\{\,,\}_\C) $, as needed.
\end{proof}
\begin{proof}[Proof of Proposition \ref{prop:cr}]
Let $V$ be the space of real-valued functions on the vertex set of the quiver $\Q$. Then $V$ carries an involution $\sigma := \tau^*$ (pull-back by $\tau$) and a Poisson bracket defined by \eqref{lcpb}, where the coordinates $y_i \colon V \to \R$ on $V$ are defined by $y_i(\xi) := \xi(i)$. %In terms of the space $V$ and involution $\tau$, the complex torus $ Y_{\Q}$ is an open dense subset of the complex vector space $V \otimes \C$, while the real part  $ Y_{\Q}^\R$ of $ Y_{\Q}$ is an open dense subset of the real vector space $V_\tau$.  
So, by Lemma \ref{lemma:cr}, the extension of the Poisson structure from $V$ to $V \otimes \C$ restricts to the space $V_\sigma = \{\xi \in V \otimes \C \mid \tau^*\xi = \bar \xi \}$. But $ Y_{\Q}$ is an open dense subset of  $V \otimes \C$, while $ Y_{\Q}^\R$  is an open dense subset of  $V_\sigma$, so the Poisson structure on $ Y_{\Q}$ restricts to $ Y_{\Q}^\R$. 
%Thus, by Lemma \ref{lemma:cr} the space $ Y_{\Q}^\R$ indeed inherits a Poisson structure from $ Y_{\Q}$.
\end{proof}

For a quiver $\Q$ with an involution $\tau$, a \textit{real quiver mutation} is either a quiver mutation at a vertex fixed by~$\tau$, or a composition of two quiver mutations at vertices switched by $\tau$. In what follows, we assume that no vertices of $\Q$ are fixed by $\tau$. In that case, a real quiver mutation is necessarily a composition of two mutations. The order of those mutations does not matter because two vertices switched by a quiver involution are necessarily disjoint.

If a quiver $\Q'$ is obtained from a quiver $\Q$ with involution $\tau$ by means of a real quiver mutation, then $\tau$ is also an involution of $\Q'$. The corresponding \textit{real $Y$-mutation} is the composition of two $Y$-mutations corresponding to quiver mutations producing $\Q'$ from $\Q$. A real $Y$-mutation commutes with the anti-holomorphic involution $\bar \tau$ and hence can be viewed a birational Poisson map  $ Y_{\Q}^\R \to  Y_{\Q'}^\R$.
%\begin{figure}
%\centering
%\begin{tikzpicture}
%\node (A) at (0,0) {$y_1$};
%\node (B) at (2,0) {$y_2$};
%\node (C) at (0,-1.5) {$\bar y_1$};
%\node (D) at (2,-1.5) {$\bar y_2$};
%\draw [->] (A) -- (B);
%\draw [<-] (A) -- (D);
%\draw [->] (C) -- (D);
%\draw [<-] (C) -- (B);
%%\draw [<->, dashed] (A) .. controls +(-0.5,-0.75) .. (C) node[midway, left] {$\tau$};
%%\draw [<->, dashed] (B) .. controls +(0.5,-0.75) .. (D)node[midway, right] {$\tau$};
%\end{tikzpicture}
%\caption{ A real mutation.}\label{Fig4}
%\end{figure}
%\begin{example}
%
%\end{example}

A real isomorphism $\Q \to \Q'$ of quivers with involutions is an isomorphism respecting the involutions. Such an isomorphism induces a real Poisson isomorphism $Y_{\Q }^\R \to Y_{\Q'}^\R $.
 A \textit{real cluster transformation} for a quiver $\Q$ with involution is a composition of real $Y$-mutations and a map induced by a real isomorphism of the resulting quiver onto the initial one. Such a transformation is a birational Poisson  map $ Y_{\Q}^\R $ onto itself.

\subsection{Recutting as a real cluster transformation}\label{subsec:cluster}

 Here we apply the formalism developed in the previous section to provide a cluster description of polygon recutting on the space $\mathcal P_n^S / S$ of similarity classes of polygons closed up to similarity. To that end we build a quiver $\Q_n$ with an involution $\tau$ such that  \eqref{eq:yrels} is a real cluster transformation as defined in Section~\ref{sec:qmrs}. In terminology of \cite{galashin2019quivers}, the quiver  $\Q_n$ is the \textit{twist} of the affine Dynkin diagram $\tilde A_{n-1}$. It has $2n$ vertices which we label as $1, \dots, n, 1', \dots, n'$. There is an arrow from vertex $i$ to vertex $j$ if and only if $j - i \equiv 1 \mod n$, an arrow from vertex $i'$ to vertex $j'$ if and only if $j - i \equiv 1 \mod n$, an arrow from vertex $i$ to vertex $j'$ if and only if $j - i \equiv -1 \mod n$, and an arrow from vertex $i'$ to vertex $j$ if and only if $j - i \equiv -1 \mod n$. 
%, with the corresponding $Y$-variables being $y_1, \dots, y_n, \bar y_1, \dots, \bar y_n$. %(which identifies the quotient ${ \mathcal P_n^{S}}/{ S}$ with an open dense subset of $Y_{\Q_n}^\R$). 
%For every $i$, there is an arrow from vertex $i$ to vertices ${i+1}$ and ${(i-1)'}$, and from vertex $ i'$ to vertices $ {(i+1)'}$ and ${i-1}$. There are no other arrows. 
The involution $\tau \colon \Q_n \to \Q_n$ is given by $\tau(i) = i'$. Since on $Y_{\Q_n}^\R$ we have $y_{i'} = \bar y_i$, we denote the $y$-variables corresponding to $i'$ vertices by $\bar y_i$. Thus, the $y_i$ and $\bar y_i$ variables are independent on $Y_{\Q_n}$ but complex conjugate to each other on $Y_{\Q_n}^\R$. Figure \ref{Fig3} depicts the quiver $\Q_5$ (while the top left part of Figure \ref{FigQR} shows the local structure of the general quiver $\Q_n$). The labels at vertices are the corresponding $y$-variables. %on  $Y_{\Q_5}^\R$ (so that the labels at vertices switched by the involution are complex conjugate to each other).

%Now recall that the space  ${ \mathcal P_n^{S}}/{ S}$ is parametrized by $n$ non-zero complex numbers $y_1, \dots, y_n \in 
%\C^*$ (Proposition \ref{prop:ycoords}). The space  $Y_{\Q_n}$ is parametrized
Since the space  $Y_{\Q_n}^\R$ is parametrized by variables $y_i \in \C^*$, Proposition \ref{prop:ycoords} gives a way to identify the space $Y_{\Q_n}^\R$ with $ \mathcal P_n^{S} \! / S$. Namely, one takes an $n$-tuple $(y_1, \dots, y_n) \in Y_{\Q_n}^\R$ and extends it by periodicity. Under this identification recutting becomes a real $Y$-mutation: 
%The corresponding space $Y_{\Q_n}^\R$ is thus parametrized by variables $y_i, y_{i'}$ subject to relations $y_{i'} = \bar y_i$. So, $y_1, \dots, y_n$ (or, more precisely, their real and imaginary parts) are coordinates in $Y_{\Q_n}^\R$.

 %Therefore, in view of Proposition \ref{prop:ycoords}, one has an identification of the space ${ \mathcal P_n^{S}}/{ S}$ with an open dense subset of $Y_{\Q_n}^\R$ given by $y_i \neq 0$ for all $i$.

\begin{figure}[b!]
\centering
\begin{tikzpicture}

\node (D1) at (0,0) {
\begin{tikzpicture}
%\node () at (-1,-0.75) {$\Q_n$};
\node (A) at (0,0) {$y_{j-1}$};
\node (B) at (2,0) {$y_{j}$};
\node (C) at (4,0) {$y_{j+1}$};
\node (D) at (0,-2) {$\bar y_{j-1}$};
\node (E) at (2,-2) {$\bar y_j$};
\node (F) at (4,-2) {$\bar y_{j+1}$};
\draw [->] (A) -- (B);
\draw [<-] (C) -- (B);
\draw [<-] (E) -- (D);
\draw [->] (E) -- (F);
\draw [->] (B) -- (D);
\draw [->] (C) -- (E);
\draw [->] (E) -- (A);
\draw [->] (F) -- (B);
\draw [<-] (A.west) -- +(-0.5,0);
\draw [->] (C.east) -- +(0.5,0);
%\draw [->] (A) -- +(-0.66,-0.5);
\draw [<-] (D.west) -- +(-0.5,0);
\draw [->] (F.east) -- +(0.5,0);
%\draw [->] (D) -- +(-0.66,0.5);
\end{tikzpicture}
};

\node (D2) at (8,0) {
\begin{tikzpicture}
\node (A) at (0,0) {$\displaystyle{y_{j-1}}(1 + y_j)$};
\node (B) at (2,0) {$y_{j}^{-1}$};
\node (C) at (4,0) {$\displaystyle \frac{y_{j+1}}{1+y_j^{-1}}$};
\node (D) at (0,-2) {$\displaystyle \frac{\bar y_{j-1}}{1+y_j^{-1}}$};
\node (E) at (2,-2) {$\bar y_j$};
\node (F) at (4,-2) {$\displaystyle{\bar y_{j+1}}(1 + y_j) $};
\draw [<-] (A) -- (B);
\draw [->] (C) -- (B);
\draw [<-] (E) -- (D);
\draw [->] (E) -- (F);
\draw [<-] (B) -- (D);
\draw [->] (C) -- (E);
\draw [->] (E) -- (A);
\draw [<-] (F) -- (B);
\draw [->] (A) -- (D);
\draw [->] (F) -- (C);
\draw [->] (A)  .. controls (1,1) and (3,1) .. (C);
\draw [->] (F)  .. controls (3,-3) and (1,-3) .. (D);
\draw [<-] (A.west) -- (-1.5,0);
\draw [->] (C.east) -- (5.5,0);
%\draw [->] (A) -- +(-0.66,-0.5);
\draw [<-] (D.west) --  (-1.5,-2);
\draw [->] (F.east) --  (5.5,-2);
%\draw [->] (D) -- +(-0.66,0.5);
\end{tikzpicture}
};

\node (D3) at (4,-5) {
\begin{tikzpicture}
%\node () at (-2,-0.75) {$\Q_n'$};
\node (A) at (0,0) {$\displaystyle\frac{y_{j-1}(1 + y_j)}{1+\bar y_j^{-1}}$};
\node (B) at (2,0) {$y_{j}^{-1}$};
\node (C) at (4,0) {$\displaystyle \frac{y_{j+1}(1+ \bar y_j)}{1+y_j^{-1}}$};
\node (D) at (0,-2) {$\displaystyle \frac{\bar y_{j-1}(1 + \bar y_j)}{1+y_j^{-1}}$};
\node (E) at (2,-2) {$\bar y_j^{-1}$};
\node (F) at (4,-2) {$\displaystyle\frac{\bar y_{j+1}(1 + y_j)}{1+\bar y_j^{-1}} $};
\draw [<-] (A) -- (B);
\draw [->] (C) -- (B);
\draw [->] (E) -- (D);
\draw [<-] (E) -- (F);
\draw [<-] (B) -- (D);
\draw [<-] (C) -- (E);
\draw [<-] (E) -- (A);
\draw [<-] (F) -- (B);
\draw [<-] (A.west) -- +(-0.5,0);
\draw [->] (C.east) -- +(0.5,0);
%\draw [->] (A) -- +(-0.66,-0.5);
\draw [<-] (D.west) -- +(-0.5,0);
\draw [->] (F.east) -- +(0.5,0);
%\draw [->] (D) -- +(-0.66,0.5);
%\draw [->] (A) -- (D);
%\draw [->] (F) -- (C);
%\draw [->] (A)  .. controls (1,1) and (3,1) .. (C);
%\draw [->] (F)  .. controls (3,-2.5) and (1,-2.5) .. (D);
%\draw [->] (A) -- (D);
%\draw [->] (F) -- (C);
%\draw [->] (A)  .. controls (1,1) and (3,1) .. (C);
%\draw [->] (F)  .. controls (3,-2.5) and (1,-2.5) .. (D);
\end{tikzpicture}
};
\draw [->, dashed] (D1) -- (D2)node[midway, above] {mutation at $y_j$};;
\draw [->, dashed] (D2) -- (D3)node[midway, right] {\,\,mutation at $\bar y_j$};;
\draw [->, dashed] (D3) -- (D1)node[midway, left] {isomorphism\,};;
\node () at (D1.west) {$\Q_n$};
\node () at (D3.west) {$\Q_n'$};
%\draw [<->, dashed] (A) .. controls +(-0.5,-0.75) .. (C) node[midway, left] {$\tau$};
%\draw [<->, dashed] (B) .. controls +(0.5,-0.75) .. (D)node[midway, right] {$\tau$};
\end{tikzpicture}
\caption{Recutting as a real cluster transformation.}\label{FigQR}
\end{figure}
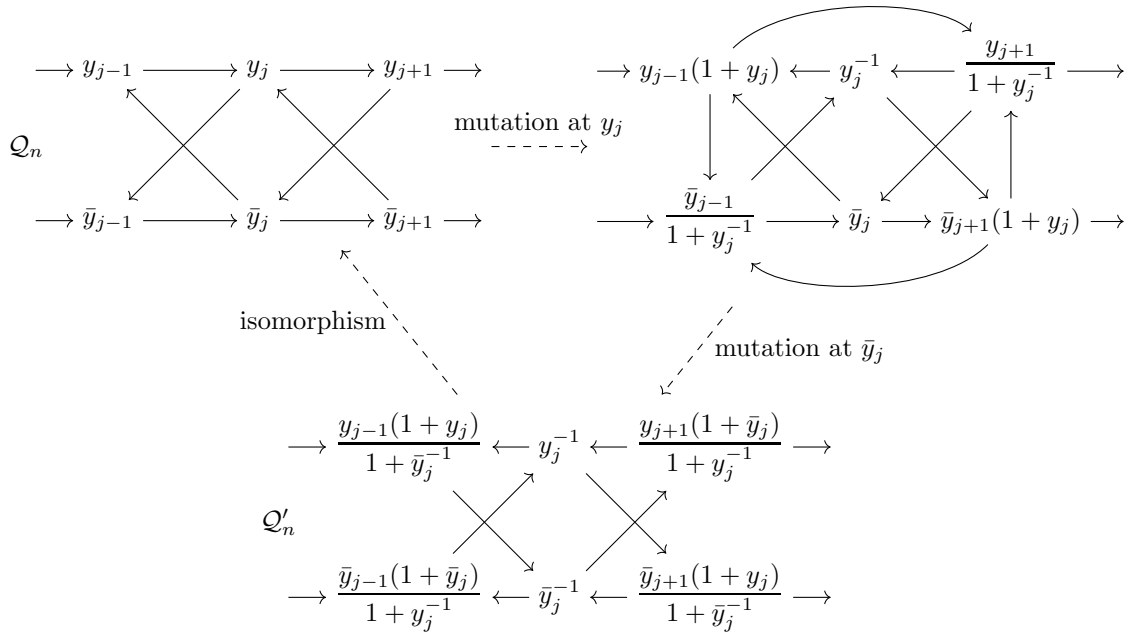

\begin{proposition}\label{prop:cluster}
Consider a real quiver mutation of $\Q_n$ given by mutating $y_j$ and $\bar y_j$. Then the resulting quiver $\Q_n'$ is real isomorphic to $\Q_n$. The cluster transformation given by composition of the real $Y$-mutation $Y_{\Q_n}^\R \to Y_{\Q_n'}^\R$ and the mapping $Y_{\Q'_n}^\R \to Y_{\Q_n}^\R$ induced by the isomorphism $\Q_n'\simeq\Q_n$ coincides with recutting $\rho_j$.
\end{proposition}
\begin{proof}
Consider Figure \ref{FigQR}. Observe that the mapping $\Q_n' \to \Q_n$  given by $y_j^{-1} \mapsto \bar y_j$ and $\bar y_j^{-1} \mapsto  y_j$ and keeping the other vertices in place is an isomorphism.  By moving the labels from $\Q_n'$ to $\Q_n$ as prescribed by the isomorphism, one gets formulas \eqref{eq:yrels}, as desired.
\end{proof}
\begin{corollary}
Recutting on the space ${ \mathcal P_n^{S}}\!/{ S}$ preserves the following Poisson bracket:
\begin{equation}\label{eq:cpb}
\begin{gathered}
\{y_i, y_{j}\} = (\delta_{i+1, j} - \delta_{i-1, j})y_iy_j, \quad \{y_i, \bar y_{j}\} = (\delta_{i-1, j} - \delta_{i+1, j})y_i \bar y_{j},\\
\{\bar y_i, \bar y_{j}\} = (\delta_{i+1, j} - \delta_{i-1, j})\bar y_i\bar y_j, \quad \{\bar y_i,  y_{j}\} = (\delta_{i-1, j} - \delta_{i+1, j})\bar y_i y_{j}.
\end{gathered}
\end{equation}
where $\delta_{i,j} = 1$ if $i \equiv j \mod n$ and $\delta_{i,j} = 0$ if $i \not\equiv j \mod n$. (Note that the last two formulas are determined by the first two since the bracket must be real.)
%All other brackets either vanish or are determined by these by skew-symmetry and reality conditions. 
%In terms of the coordinates $|y_i|$ and $\phi_i := \mathrm{arg}(y_i)$ (i.e., ratios of consecutive side lengths and exterior angles), the bracket takes the form $\{\phi_i, \phi_{i+1}\} = 1$, while $|y_i|$ are Casimirs.
\end{corollary}
\begin{proof}
This is the canonical Poisson bracket on $Y_{\Q_n}^\R$, preserved by all real cluster transformations.
\end{proof}
\begin{remark}\label{rem:phibr}
Brackets \eqref{eq:cpb} take a particular nice form when written in terms of  $|y_i|$ and $\phi_i = \mathrm{arg}(y_i)$, i.e. ratios of lengths of consecutive sides and exterior angles of the polygon. Namely, $|y_i|$ are Casimirs, while \begin{equation}\label{eq:phibr}\{\phi_i, \phi_j\} = \delta_{i+1, j} - \delta_{i-1, j}.\end{equation}
\end{remark}
\begin{remark}\label{rem:manyCas}
The bracket \eqref{eq:cpb} has a large number of Casimirs, namely all $|y_i|$, the product $y_1 \cdots y_n$ (equal to the coefficient $\alpha \in \C^*$ of the monodromy transformation $z \mapsto \alpha z + \beta$), and, for even $n$, the product $y_1y_3 \cdots y_{n-1}$. Most of these Casimirs are not preserved by recutting transformations  \eqref{eq:yrels}. The only ones that are preserved are the function $y_1 \cdots y_n$ (in particular, the angle sum $\sum \phi_i = \mathrm{arg}(y_1 \cdots y_n)$) and, for even $n$, the function $\mathrm{arg}(y_1y_3 \cdots y_{n-1}) = \phi_1 + \phi_3 + \dots + \phi_{n-1}$.

It follows from this description of Casimirs that the spaces $\mathcal P_n^E / S$, $\mathcal P_n^T / S$ of similarity classes of polygons closed up to isometry or translation are Poisson submanifolds. Indeed, the defining equation of $\mathcal P_n^E / S$ inside $\mathcal P_n^S / S$ is $|y_1 \cdots y_n| = 1$, while the defining equation of $\mathcal P_n^T / S$ is $y_1 \cdots y_n = 1$. So both are level sets of Casimirs and hence Poisson submanifolds. As for the submanifold $\mathcal P_n / S$ of similarity classes of closed polygons, it is defined by equations
$
y_1 \cdots y_n = 1$ and $1 + y_1 + y_1y_2 + \dots + y_1 \cdots y_{n-1} = 0,
$
and is therefore not Poisson.
%and the function $\{y_n,1 + y_1 + y_1y_2 + \dots + y_1 \cdots y_{n-1} \} $
\end{remark}
\begin{remark}
The cluster transformation described in Proposition \ref{prop:cluster} can be regarded a particular case of a more general transformation, known as the \textit{geometric $R$-matrix}. The cluster geometric $R$-matrix is defined in~\cite{inoue2019cluster} for triangular grid quivers and in \cite{chepuri2020plabic} for a more general class of \textit{spider web quivers}. The quiver $\Q_n$ is not a triangular grid quiver or spider web quiver but can be seen as such if we add obsolete arrows from each $y_i$ to $\bar y_i$ and from each $\bar y_i$ to $y_i$.
\end{remark}

%\begin{lemma}
%
%Assume that a polygon $(v_j')$  is obtained from $(v_j)$ by recutting $\rho_i$. Let  $z_j := v_j - v_{j-1}$ and  $z_j' := v_j' - v_{j-1}'$. Then
%%
%%\begin{gather*}
%%z_j' = \bar z_{j+1}\frac{z_j + z_{j+1}}{\bar z_j + \bar z_{j+1}},\\
%%z_{j+1}' = \bar z_j\frac{z_j + z_{j+1}}{\bar z_j + \bar z_{j+1}}.
%%\end{gather*}
%\end{lemma}

% 
% Conversely, given an $n$-periodic sequence $y_i$, one can reconstruct an $n$-gon $(v_i)$, closed up to a complex affine transformation, such that
% $$
% \frac{v_{i+1} - v_i}{v_i - v_{i-1}} = y_i.
% $$
% \begin{definition}
% Let $\phi \colon \C$. A \textit{planar $n$-gon closed up to an element of $G$} is a bi-infinite sequence $(v_i \in \C)_{i \in \Z}$ such that $v_i \neq v_{i+1}$ for all $i \in \Z$, and there exists $\phi \in G$ such that $v_{i+n} = \phi(v_i)$ for all $i \in \Z$. The element $\phi \in G$ is called the monodromy of the polygon $(v_i)$.
% \end{definition}
% Assume that $(v_i)$ is an $n$-gon 
%  \begin{definition}
% Let $G$ be a subgroup of the group of complex-affine transformations $z \mapsto \alpha z + \beta$. For an $n$-gon $(v_i)$ closed up to an element of $G$, \textit{recutting} $\rho_i$ is defined as reflection of each vertex $v_j$ with $j \equiv i \mod n$  in the perpendicular bisector of the interval $(v_{j-1}, v_{j+1})$.
% \end{definition}
%R%ecutting $\rho_i$ is defined for all $n$-gons 

\section{Recutting of polygons closed up to isometry: Arnold-Liouville integrability}\label{sec:ial}
 \subsection{Quaternionic polynomials}\label{subsec:qp}
This section is a brief introduction into the theory of polynomials over quaternions. We begin by reviewing their general properties. All these results are well known but seem to be scattered in the literature, so we sketch proofs. We then move on to define what we call \textit{special quaternionic polynomials}  and present a criterion for factorization of such polynomials into linear factors. This result plays an instrumental role in our proof of integrability of recutting.\par

First, let us fix some terminology. Let $\H = \mathrm{span}_\R\langle 1, \i, \j, \k \rangle$ be the skew-field of quaternions. There are two different operations in $\H$ that are usually referred to as \textit{conjugation}: $\alpha = a + b \i + c\j + d \k \mapsto \bar\alpha = a - b \i - c\j - d \k$ and $\alpha \mapsto \beta \alpha \beta^{-1}$. To avoid confusion, we only use the term ``conjugation'' for the former operation. Quaternions of the form $\alpha$, $\beta \alpha \beta^{-1}$ will be called \textit{similar}. %The same applies to 

 Let $\H[t] := \H \otimes \R[t]$ be the $\R$-algebra of \textit{unilateral}  quaternionic polynomials in the indeterminate $t$. Those can be thought as polynomials over quaternions whose coefficients commute with the variable. For 
$
f = \sum_{i=0}^n \alpha_i t^i \in \H[t]$ (where $\alpha_i \in \H$) and a quaternion $\beta \in \H$, define the \textit{right evaluation of $f$ at $\beta$} as $\mathrm{ev}^r_f (\beta) := \sum_{i=0}^n \alpha_i \beta^i$,  and 
 \textit{left evaluation of $f$ at $\beta$} as
$
\mathrm{ev}^l_f (\beta) := \sum_{i=0}^n  \beta^i\alpha_i
$
(for real $\beta$ one has $\mathrm{ev}^r_f (\beta)= \mathrm{ev}^l_f (\beta)$, in which case we just write it as $f(\beta)$). Say that $\beta$ is a right (left) root of $f$ if $\mathrm{ev}^r_f (\beta) = 0$ (respectively, $\mathrm{ev}^l_f (\beta) = 0$).  
The following summarizes basic facts about quaternionic polynomials and their roots.
\begin{proposition}\label{prop:qp} %{\color{red} root iff divisible by ...}
\begin{enumerate}
\item Let $f, g \in \H[t]$ and let $\alpha \in \H$ be a right root of $g$. Then $\alpha$ is a right root of the product $fg$.
\item Let $f \in \H[t]$ and let $\alpha \in \H$. Then $\alpha$ is a right root of $f$ if and only if $t - \alpha \in \H[t]$ is a right divisor of $f$.% gordon1965zeros
\item Let $f, g \in \H[t]$ and assume that $\alpha \in \H$ is not a right root of $g$. Then
\begin{equation}\label{eq:evfg}
\mathrm{ev}^r_{fg} (\alpha) = \mathrm{ev}^r_{f} ( \mathrm{ev}^r_g (\alpha)\cdot \alpha\cdot \mathrm{ev}^r_g (\alpha)^{-1}) \mathrm{ev}^r_g (\alpha).
\end{equation}
\item  Let $f \in \H[t]$ and let $[\alpha] := \{ \beta \alpha \beta^{-1} \mid \beta \in \H \setminus \{0\}\} \subset \H$ be a similarity class of quaternions. Then one of the following is true:
\begin{enumerate}
\item The class $[\alpha]$ contains neither right nor left roots of $f$.
\item The class $[\alpha]$ contains a unique right root of $f$ and a unique left root of $f$. Such roots are called \textbf{isolated}.
\item Any element of $[\alpha]$ is both right and left root of $f$. Such roots are called \textbf{spherical}. 
\end{enumerate}
\item A similarity class $[\alpha] \subset \H$ contains a root of $f \in \H[t]$ if and only if it contains the root of its \textbf{companion polynomial} $f \bar f = \bar f f \in \R[t]$. In particular, any non-zero quaternionic polynomial has at least one root, and hence can be factored into linear factors. (By the previous part, here we do not need to distinguish between right and left roots.)
\item A similarity class  $[\alpha] \subset \H$ consists entirely of roots of $f$ if and only if $f$ is divisible by the \textbf{characteristic polynomial} of $\alpha$, given by
\begin{equation}\label{eq:charp}
\chi_\alpha = t^2 - 2(\mathrm{Re} \,\alpha) t + |\alpha|^2.
\end{equation}
(Note that since $\chi_\alpha$ has real coefficients, divisibility of $f$ by $\chi_\alpha$ on the right is equivalent to divisibility of $f$ by $\chi_\alpha$ on the left).
\item For any non-zero $f \in \H[t]$, the total number of similarity classes  $[\alpha] \subset \H$ containing roots of $f$ does not exceed the degree of $f$.
\end{enumerate}
\end{proposition}
\begin{proof}
Parts 1-3 hold for polynomials over any division ring, cf. \cite[Theorem 1]{gordon1965zeros}. To prove part 1 we need  to establish the implication $\mathrm{ev}^r_g (\alpha) = 0\implies\mathrm{ev}^r_{fg} (\alpha) = 0$. Since the map  $\H[t] \to \H$ given by $f \mapsto \mathrm{ev}^r_{fg} (\alpha)$ is a homomorphism of left $\H$-modules, it suffices to consider the case $f = t^m$. For such $f$ we have $fg = t^m g = gt^m$, so
$
\mathrm{ev}^r_{fg} (\alpha) = \mathrm{ev}^r_g (\alpha)\alpha^m = 0,
$
as needed. Now that we established part~1, part 2 can be proved in the same way as for a field, using long division. To prove part 3, let $\beta := \mathrm{ev}^r_g (\alpha)$. Again, it suffices to consider the case $f = t^m$. In that case, we get 
$
\mathrm{ev}^r_{fg} (\alpha) = \beta \alpha^m = (\beta \alpha \beta^{-1})^m \beta = \mathrm{ev}^r_{f}(\beta \alpha \beta^{-1})\beta,
$
as desired. \par To prove part 4, observe that any quaternion $\alpha$ is a root of its {characteristic polynomial} \eqref{eq:charp}. From this it follows that any positive power of a quaternion $\alpha$ can be expressed as $\alpha^m = r + s \alpha$, where  $r, s$ are polynomials with real coefficients in terms of $\mathrm{Re} \,\alpha$ and $|\alpha|^2$, and in particular only depend on the similarity class of $\alpha$. This in turn implies that for any polynomial $f \in \H[t]$ and a similarity class $[\alpha] \subset \H$ there exist $\lambda, \mu \in \H$ such that for any $\alpha' \in [\alpha]$ we have
$
\mathrm{ev}^r_{f} (\alpha') = \lambda + \mu \alpha'$ and $\mathrm{ev}^l_{f} (\alpha') = \lambda +  \alpha' \mu.
$
Now it is easy to see that (a) holds when $|\lambda| \neq |\alpha| |\mu|$,  (b) holds when $|\lambda| = |\alpha| |\mu| \neq 0$, while possibility~(c) holds when $|\lambda| = |\alpha| |\mu| = 0$.
\par
To prove part 5, notice that by part 1 any right root of $f$ is also a right root of $\bar f f$. So it suffices to show that if $\alpha$ is a right root of $\bar f f$, then the similarity class of $\alpha$ contains a root of $f$. Assume that $\alpha$ is a right root of~$\bar f f$. If $\alpha$ is also a right root of $f$, then we are done. If not, then by \eqref{eq:evfg} we get that $\alpha' : = \mathrm{ev}^r_{f}(\alpha) \cdot \alpha \cdot \mathrm{ev}^r_f(\alpha)^{-1}$ is a right root of $\bar f$, which is equivalent to saying that $\bar \alpha'$ is a left root of $f$. But any quaternion is similar to its conjugate (since conjugation preserves the real part and absolute value, and two quaternions $\alpha, \beta \in \H$ are similar if and only if $\mathrm{Re}\, \alpha =\mathrm{Re}\,  \beta$ and $|\alpha| = |\beta|$), so $\bar \alpha'$ is similar to $\alpha'$ and hence to $\alpha$. So indeed $f$ has a root in the similarity class of $\alpha$, as needed.
%Now observe that if $\alpha$ and $\beta$

To prove part 6 note that any element of the class $[\alpha]$ is a root of $\chi_\alpha$. %The rest of the proof is standard and is based on division of $f$ by $\chi_\alpha$ with a remainder.%
So the class $[\alpha]$ is annihilated by $f$ if and only if is annihilated by the remainder of right division of $f$ by $\chi_\alpha$. But that remainder is at most linear, so it can only annihilate the class $[\alpha]$ if it vanishes, which means that $f$ is divisible by $\chi_\alpha$, as needed. 
%So if $\chi_\alpha$ is a right divisor $f$, then any element of the class $[\alpha]$ is also a root of $f$, as needed. Conversely, assume that any element of $[\alpha]$ is a root of $f$. Division of $f$ by $\chi_\alpha$ gives a remainder which is at most linear. This remainder must annihilate the whole class  $[\alpha]$ and hence vanish, as desired.

%Then one can use the division algorithm to write $p$ as $ p = sq + r$ where the remainder $r$ is linear. But since the remainder $r$

Part 7 is also true for any division ring, see \cite[Theorem 2]{gordon1965zeros}. Let us sketch a quaternion-specific proof. By part 5 it suffices to show that the number of similarity classes containing roots of the companion polynomial~$\bar f f$ does not exceed the degree $n$ of $f$. To that end observe that since the polynomial~$\bar f f$ has real coefficients, the similarity class of any of its roots consists entirely of roots. So, any similarity class containing a root of  $\bar f f$ contains a complex root of ~$\bar f f$. Therefore it suffices to show that the number of similarity classes of complex roots of $\bar f f$ is at most $n$. %To prove that we will show that $\alpha$ is a root of $\bar p p$, then either the 
%Therefore it suffices to show that the total number of conjugacy classes  $[\alpha] \subset \H$ containing complex roots of the companion polynomial $\bar p p$ does not exceed $n$. 
To prove that write $f$ as $f_1 + f_2 \i + f_3 \j + f_4 \k$, where the polynomials $f_i$ are real. Then~$\bar f f = \sum f_i^2$ so all its real roots have multiplicity at least $2$. As for non-real roots, any such root $\alpha$ has its complex conjugate counterpart $\bar \alpha$ which is similar to $\alpha$. So any  similarity class  of complex roots of $\bar f f$ contains at least two roots (counted with multiplicity), and the total number of classes cannot exceed $n$, q.e.d.
\end{proof}
%The following result gives a description of roots for quaternionic quadratic equations. For simplicity, we only consider the case of non-real coefficients. 
%\begin{proposition}
%Any quadratic quaternionic polynomial $p(t) = t^2 + \alpha t + \beta$, where $a, b \in \H$ are not both real, has at most two left roots and two right roots. 
%\end{proposition}
%\begin{proof}
%By the fifth statement of Proposition \ref{prop:qp}
%%By part 4 of the previous proposition, $p(t)$ has at least one right root $b \in \H$. Thus one can write $p$ as
%%$p(t) = (t-a)(t-b)$.
%\end{proof}

\begin{definition}
We say that a quaternionic polynomial $f \in \H[t]$ is \textit{special} if it satisfies one of the following equivalent conditions:
\begin{enumerate}
\item
$
f(-t) = \i f(t) \i^{-1}.% = -\i f(t) \i.
$
\item $f$ can be written as a polynomial in $\j t$ with complex coefficients.
\item  All even coefficients of $f$ are complex numbers, while all odd coefficients belong to the complementary subspace $\mathrm{span}_\R\langle \j, \k \rangle$. \end{enumerate}
\end{definition}
Special quaternionic polynomials form a subalgebra of $\H[t]$ which we denote by $\tilde \H[t]$. More generally, one can take an arbitrary non-zero quaternion $\alpha$ with zero real part and consider polynomials such that 
$f(-t) = \alpha f(t) \alpha^{-1}$. This always gives a subalgebra isomorphic to $\tilde \H[t]$. 
The following property of special quaternionic polynomials will be used to prove integrability of polygon recutting:
\begin{proposition}\label{prop:facspec}
A special quaternionic polynomial $f \in \tilde \H[t]$ can be written as a product of linear special quaternionic polynomials $f_i \in  \tilde \H[t]$ if and only if all complex roots of the companion polynomial $\bar f f$ of $f$ are on the imaginary axis. 
\end{proposition}

We first prove a lemma, which will also be useful by itself.

\begin{lemma}\label{lemma:isroot}
Assume that the companion polynomial $\bar f f$ of a special quaternionic polynomial $f \in \tilde \H[t]$ has all its complex roots on the imaginary axis. Let $\alpha$ be a root of $f$. Then:
\begin{enumerate}
\item There exists $\beta \in \mathrm{span}_\R\langle \j, \k \rangle$ which is similar to $\alpha$.
\item Moreover, if $\alpha$ is isolated, then $\alpha \in \mathrm{span}_\R\langle \j, \k \rangle$.
\end{enumerate}%, while any spherical root of $f$ is similar to an element of $\mathrm{span}_\R\langle \j, \k \rangle$.
\end{lemma}
\begin{proof}[Proof of Lemma \ref{lemma:isroot}]
%Let $\alpha$ be a root of $f$. Then, 
Since all roots of the companion polynomial of $f$ are on the imaginary axis, by part 5 of Proposition \ref{prop:qp} we have $\mathrm{Re}\, \alpha = 0$. So, $\alpha$ must be similar to some element of $\mathrm{span}_\R\langle \j, \k \rangle$, which establishes the first statement of the lemma. To prove the second statement, assume that $\alpha$ is isolated.
Using part 2 of Proposition \ref{prop:qp} write $f(t)$  as $ g(t) (t - \alpha)$ where $g(t) \in \H[t]$. Then, using that $f$ is special, we get
%$$
%f(-t) = g(-t)(-t - a) = -g(-t)(t + a).
%$$
%On the other hand,
$$
f(t) = \i f(-t) \i^{-1} =  \i g(-t)(-t-\alpha) \i^{-1} =  -\i g(-t) \i^{-1}\cdot (t + \i \alpha \i^{-1}),
$$
and applying once again part 2 of Proposition \ref{prop:qp} we see that $\alpha' := - \i \alpha \i^{-1}$ is also a root of $f$.  Furthermore, since $\mathrm{Re}\, \alpha = 0$, we have $\mathrm{Re}\, \alpha' = 0$, and since $|\alpha'| = |\alpha|$, it follows that $\alpha'$ is similar to $\alpha$. Therefore, since $\alpha$ is isolated, we have $\alpha'= - \i \alpha \i^{-1} = \alpha$, which is equivalent to $\alpha \in \mathrm{span}_\R\langle \j, \k \rangle$, as needed.
\end{proof}

\begin{proof}[Proof of Proposition \ref{prop:facspec}]
For a linear  special quaternionic polynomial $a + b\j t$, where $a,b \in \C$, its companion polynomial $a\bar a + b \bar b t^2$ has roots on the imaginary axis. Furthermore, since the companion polynomial of a product is the product of companion polynomials, it follows that if $f$ is a product of linear special quaternionic polynomial, then all complex roots of the companion polynomial of $f$ are on the imaginary axis.  Conversely, assume that $f$ is special and all roots of $\bar f f$ are on the imaginary axis. Let $\alpha \in \H$ be an arbitrary root of $f$. Then $\alpha$ is either isolated or spherical.  In the former case, by Lemma \ref{lemma:isroot}, we have $\alpha \in \mathrm{span}_\R\langle \j, \k \rangle$, so $f$ is divisible on the right by the special quaternionic polynomial $\alpha^{-1}t - 1$. %In the latter case, we have \ref{lemma:isroot}
In the later case, by Lemma \ref{lemma:isroot} we can find a root $\alpha' \in \mathrm{span}_\R\langle \j, \k \rangle$ of $f$ similar to $\alpha$, so $f$ is divisible on the right by the special quaternionic polynomial $(\alpha')^{-1}t - 1$. So, in either case, $f$ is divisible on the right by a  linear special quaternionic polynomial, and proceeding by induction one shows that $f$ can be written as a product of such polynomials. 
\end{proof}

\subsection{Recutting as refactorization}\label{subsec:refac}
In this section we establish a connection between recutting and quaternionic polynomials, which is then used to derive invariants of recutting and prove its complete integrability. Let $(v_i)$ be a polygon, and $(v_i')$ be the result of its recutting at a vertex $v_j$. Consider the edge vectors $z_i = v_i - v_{i-1}$ and $z_i' = v_i' - v_{i-1}'$. Then $z_j, z_{j+1}, z_j', z'_{j+1}$ satisfy relations  \eqref{eq:zrels}.
%Let $\H = \mathrm{span}_\R\langle 1, \i, \j, \k \rangle$ be the skew-field of quaternions, and let $\H[t] := \H \otimes \R[t]$ be the algebra of \textit{unilateral  quaternionic polynomials} in the indeterminate $t$. The connection between recutting and this algebra is based on the following:
\begin{proposition}\label{prop:rqp}
Recutting relations \eqref{eq:zrels} are equivalent to the following relation between special quaternionic polynomials:
\begin{equation}\label{eq:refac}
(1 + z_j {\j}t)(1 + z_{j+1} {\j}t) = (1 + z'_j {\j}t)(1 + z'_{j+1} {\j}t).
\end{equation}
\end{proposition}
\begin{proof}
Indeed, for any $z, w \in \C$ one has
$(1 + z {\j}t)(1 + w {\j}t) = (1 + (z+w){\j}t - z\bar w t^2),
$
so \eqref{eq:refac} is equivalent to  \eqref{eq:zrels}.
\end{proof}
As a result, one can interpret recutting of a polygon $(v_i)$ at a vertex $v_j$ as refactorization of the quadratic quaternionic polynomial $g(t) := (1 + z_j {\j}t)(1 + z_{j+1} {\j}t)$. 
\begin{remark}\label{rem:rqp}
Note that:
\begin{enumerate}
\item The polynomial $g(t)$ is only divisible by a real polynomial when $z_j = -z_{j+1}$ (equivalently,  $v_{j-1} = v_{j+1}$), which is the case when recutting at $v_j$ is impossible. So, as long as recutting is possible, it follows from Proposition \ref{prop:qp} that the polynomial $g(t)$ has at most two right roots (both isolated) and hence at most two factorizations of the form $(1 + z {\j}t)(1 + w {\j}t) $. \item The companion polynomial $\bar g(t) g(t)$ of $g(t)$ is $(1 + |z_j|^2t^2)(1 + |z_{j+1}|^2t^2)$. So, if $|z_j| \neq |z_{j+1}|$, the polynomial $g(t)$ has exactly two right roots (both isolated) and hence, by Lemma~\ref{lemma:isroot}, exactly two factorizations of the form $(1 + z {\j}t)(1 + w {\j}t) $. In this case, recutting at $v_j$ can be seen as switching between these two factorizations.
\item If $|z_j| = |z_{j+1}|$, then recutting at $v_j$ is the identity transformation. In this case, the polynomial $g(t)$ has a unique right root and hence a unique factorization. 
\end{enumerate}
Summing up, unless the vertices $v_{j-1}$ and $v_{j+1}$ coincide, the polynomial $g(t) = (1 + z_j {\j}t)(1 + z_{j+1} {\j}t) $ has two (possibly identical) factorizations of the form $(1 + z {\j}t)(1 + w {\j}t) $, and recutting can be thought as switching from one factorization to another.
\end{remark}

\subsection{Recutting invariants for polygons closed up to translation}\label{subsec:inv1}
We begin our description of recutting invariants with the case of polygons closed up to translation. In this case, the invariants have a particularly simple form. In the next section generalize these results to polygons closed up to isometry.

Given a polygon $(v_i) \in \mathcal P_n^{T}$ closed up to translation, let $z_i$ be its edge vectors. Consider a special quaternionic polynomial
\begin{equation}\label{eq:fpoly}
f(t) := (1 + z_1 {\j}t)\cdots(1 + z_{n} {\j}t) \in \tilde \H[t].
\end{equation}
%Then all polynomials $F_i(t)$ are conjugate to $F(t)$ in the skew-field $\H \otimes \R[[t]]$. Let $[F(t)]$ be the corresponding conjugacy class. 
\begin{proposition}\label{prop:Lax}
The similarity class of the polynomial $f(t)$ in the skew-field $\tilde \H[[t]]$ is invariant under both the action of the group $E$ of isometries, and the recutting action of $\tilde S_n$.
\end{proposition}
\begin{remark}
The skew-field  $\tilde \H[[t]]$ of special quaternionic power series is defined analogously to $\tilde \H[t]$: a quaternionic power series is special if and only if all its even coefficients are complex numbers, while all odd coefficients belong to the complementary subspace $\mathrm{span}_\R\langle \j, \k \rangle$.  We say that $f,g \in \tilde \H[t]$ are similar if there exists an invertible formal power series $h \in \tilde \H[[t]]$ such that $hfh^{-1} = g$.
\end{remark}
\begin{proof}[Proof of Proposition \ref{prop:Lax}]
The action of the group $E$ of isometries amounts to multiplying all $z_i$ by the same complex number $\alpha$ of absolute value $1$. This is equivalent to a similarity transformation $f \mapsto \alpha^{1/2} f \alpha^{-1/2}$. So the action of $E$ indeed preserves the similarity class of $f$.\par
To prove the invariance of the similarity class of $f$ under recutting, observe that by Proposition~\ref{prop:rqp} the polynomial 
$
f_i(t) := (1 + z_i {\j}t)\dots(1 + z_{i+n - 1} {\j}t)
$
 does not change under recutting $\rho_i$. Furthermore, due to $n$-periodicity of the sequence $z_j$, the polynomial $f(t)$ is similar to $f_i(t)$, so its similarity class is preserved by any recutting $\rho_i$ and hence by the whole recutting group.
 \end{proof}
It follows from Proposition \ref{prop:Lax} that any central function of $f(t)$ descends to the space $\mathcal P_n^{T} / {E}$ and is invariant under recutting action on that space. As such functions we take the coefficients of the real polynomials $\bar f(t) f(t)$ and $\mathrm{Re}\, f(t)$.

\begin{proposition}\label{prop:inv}
For a polygon closed up to translation, one has
\begin{gather}
\bar f(t) f(t)= \prod_{i} (1 + |z_i|^2t^2) = 1 + \sum_k E_k t^{2k} ,\\
\mathrm{Re}\, f(t) = 1 + \sum_{k = 1}^{\lfloor n/2 \rfloor} (-1)^k I_{k}t^{2k}, \label{eq:trinv0}
\end{gather}
where
\begin{equation}\label{eq:trinv}
\begin{gathered}
E_k :=\sum_{i_1 <  \dots < i_{k}} |z_{i_1}|^2 \dots |z_{i_k}|^2,  \\
I_{k} := \mathrm{Re}\!\!\!\sum_{i_1 <  \dots < i_{2k}} z_{i_1}\bar z_{i_2} \dots z_{i_{2k - 1}} \bar z_{i_{2k}}.
\end{gathered}
\end{equation}
 Here and in the rest of this section all summation indices run from $1$ to $n$ unless otherwise specified.
\end{proposition}
\begin{proof}
The first equality follows from multiplicativity of the companion polynomial, while the second one is obtained by a straightforward computation.
\end{proof}
It follows that the recutting action on $n$-gons closed up to translation has $\lfloor 3n/2\rfloor$ invariants, namely $n$ elementary symmetric polynomials $E_1, \dots, E_n$ of squared side lengths $|z_i|^2$ (whose invariance is obvious from the geometric definition of recutting), and $\lfloor n/2 \rfloor$ additional invariants $I_1, \dots, I_{\lfloor n/2 \rfloor}$. %It is easy to see from the explicit form of those invariants that they a \mathrm{E}
The following result explains the geometric meaning of some of the invariants $I_{k}$:
\begin{proposition}\label{prop:geom}
%Consider a polygon closed up to a translation $z \mapsto z + \beta$. Then:
\begin{enumerate}
\item For a polygon with monodromy $z \mapsto z + \beta$, the invariant $I_1$ is a function of squared side lengths and squared length of $\beta$: 
$
I_1 = \frac{1}{2} (|\beta|^2 - E_1).
$
In particular, a polygon is closed if and only if \begin{equation}\label{eq:cc}I_1 = -\frac{1}{2}  E_1.\end{equation}
\item Let $n$ be even. Then
\begin{equation}\label{eq:cyclic}
I_{n/2} = \sqrt{E_n}\cos(\phi_1 + \phi_3 + \dots + \phi_{n-1}) =  \sqrt{E_n}\cos(\phi_2 + \phi_4 + \dots + \phi_{n}),
\end{equation}
where $\phi_i$ are exterior angles of the polygon.
\item For closed polygons, the invariant $I_2$ is a function of squared side lengths and the area $A$ of the polygon: 
\begin{equation}\label{eq:area}
I_2 = \frac{1}{2}E_2 - \frac{1}{8}E_1^2 - 2A^2.
\end{equation}

\end{enumerate}
\end{proposition}
\begin{remark}
Note that for $n = 3$ we have $I_2 = 0$, so relation \eqref{eq:area} becomes $A^2 =\frac{1}{4}E_2 - \frac{1}{16}E_1^2  $ which is nothing but Heron's formula for the area of a triangle. When $n = 4$, formulas \eqref{eq:cyclic} and \eqref{eq:area} combined together give $A^2 =\frac{1}{4}E_2 - \frac{1}{16}E_1^2  -\frac{1}{2}\sqrt{E_n}\cos(\phi_1 + \phi_3) $ which is equivalent to Bretschneider's formula for the area of a quadrilateral.
\end{remark}
\begin{proof}[Proof of Proposition \ref{prop:geom}]
%To prove the first statement, note that
%$$
%\beta = p_{n+1} - p_1 = \sum_{i=1}^n (v_{i+1} - v_i) =  \sum_{i=1}^n z_i,
%$$
%so
%We have
%$$
%\beta \bar \beta = \left(\sum_{i} z_i\right) \left(\sum_{j} \bar z_j\right) = \sum_{i,j} z_i \bar z_j =\left( \sum_{i} |z_i|^2 \right) +  \left(\sum_{i < j} (z_i \bar z_j + \bar z_i z_j)\right) =  \sum_{i} |z_i|^2 + 2I_2,
%$$
%proving the first statement. 
The first two parts are proved by a straightforward computation, so we only prove the last part. Denote by $\alpha_k$ be the coefficient of $t^k$ of the polynomial $f(t)$. Then
$
\alpha_1 = \j\sum_{i} z_i$, $\alpha_2 = -\sum_{i < j} z_i \bar z_j .
$
For closed polygons, this gives $\alpha_1  = 0$, 
$\alpha_2 =  -I_1 - 2A\i,
$
%$$
%F(t) = 1 +  \left(\sum_{i=1}^n z_i\right)\j t - \left(\sum_{1\leq i < j \leq n} z_i \bar z_j \right)t^2 + O(t^3) =  1 + (I_2 + 2A)t^2 + O(t^3),
%$$
where \begin{equation}A =\frac{1}{2}\,\mathrm{Im} \sum_{ i < j } z_i \bar z_j \end{equation} is the signed area.  Using also that $f \in \tilde \H[t]$ and so $\mathrm{Re} \,\alpha_k = 0$ for any odd $k$, we get
\begin{equation}\label{eq:ft2}
\begin{gathered}
\bar f(t) f(t)= 1 + 2(\mathrm{Re}\,\alpha_2)t^2 + 2(\mathrm{Re}\,\alpha_4 + \alpha_2\bar \alpha_2)t^4 + O(t^6)  =  1 - 2I_1 t^2 + (2I_2 + I_1^2 +4A^2)t^4 + O(t^6).
\end{gathered}
\end{equation}
So, by definition of $E_2$ as the coefficient of $t^4$ in this expansion, we have $E_2 = 2I_2 + I_1^2 +4A^2$. Combined with~\eqref{eq:cc}, this gives the desired formula.
%At the same time, by Proposition \ref{prop:inv} we have
%\begin{equation}\label{eq:ft22}
%\bar f(t) f(t)  =  \prod_{i} (1 + |z_i|^2t^2) = 1 + \left( \sum_{i} |z_i|^2 \right)t^2 + \left(\sum_{i < j } |z_i|^2|z_j|^2 \right) t^4 +  O(t^6).
%\end{equation}
%Comparing the coefficients of $t^4$ in \eqref{eq:ft2} and \eqref{eq:ft22}, one gets the desired formula for $I_4$.
\end{proof}

\begin{remark}
It follows from Proposition \ref{prop:Lax} that invariants $I_{k}$ are well-defined on the quotient $\mathcal P_n^{T} \! / {E}$, i.e.  are invariant under simultaneous rotation of all $z_i$. This is also easy to see from the explicit form of those invariants.
\end{remark}

\begin{remark}\label{rem:byc}
Let us show that our invariants  $I_k$ coincide with invariants $c_{2k}$ constructed in \cite[Proposition 4.3]{tabachnikov2012discrete}. Consider a representation $  \H \to GL_2(\C)$ given by
$$
\i \mapsto \left(\begin{array}{cc}0 & 1 \\-1 & 0\end{array}\right), \quad \j \mapsto \left(\begin{array}{cc}-\i & 0 \\0 & \i\end{array}\right), \quad \k \mapsto \left(\begin{array}{cc}0 & \i \\ \i & 0\end{array}\right).
$$
The image of the polynomial $f(t)$ given by \eqref{eq:fpoly} under this representation is the matrix polynomial
$$
F(t) = \left(\begin{array}{cc}1 - a_1 t \cos(\psi_1) \i & a_1 t \sin(\psi_1)\i \\ a_1 t \sin(\psi_1)\i & 1 + a_1 t \cos(\psi_1) \i\end{array}\right) \cdots \left(\begin{array}{cc}1 - a_n t \cos(\psi_n) \i & a_n t \sin(\psi_n)\i \\ a_n t \sin(\psi_n)\i & 1 + a_n t \cos(\psi_n) \i\end{array}\right),
$$
where $a_i := |z_i|$ and $\psi_i := \mathrm{arg}(z_i)$. So
\begin{equation}\label{eq:retr}
\mathrm{Re}\, f(t) = \frac{1}{2} \mathrm{Tr}\, F(t) = \frac{1}{2}\ell^{-n} \mathrm{Tr}\, M_1 \cdots M_n, 
\end{equation}
where $\ell := (-t \i)^{-1}$ and
$$
M_i :=  \left(\begin{array}{cc} \ell + a_i \cos(\psi_i) & -a_i \sin(\psi_i) \\ -a_i \sin(\psi_1) & \ell- a_i  \cos(\psi_i) \end{array}\right).
$$
The invariants $c_k$ of \cite{tabachnikov2012discrete} are defined by the relation
$$
\mathrm{Tr}\, M_1 \cdots M_n = 2(\ell^n + c_2 \ell^{n-2} + c_4 \ell^{n-4} + \dots),
$$
so by \eqref{eq:retr} we have
$$
\mathrm{Re}\, f(t) = 1 + c_2 \ell^{-2} + \dots = 1 - c_2 t^2 + c_4t^4 - \dots,
$$
and hence $I_k = c_{2k}$.
% So they coincide with ours up to a factor $\pm \frac{1}{2}$.
\end{remark}

%\begin{remark}
%The construction of this section generalizes to polygons whose monodromy is an orientation-preserving isometry $z \mapsto \alpha z + \beta$, where $|\alpha| = 1$. In that case, the polynomial $f$ is defined as

\subsection{Recutting invariants of polygons closed up to isometry}\label{sec:inveuc}
In the previous section we constructed recutting invariants on the space  $\mathcal P_n^{T}$ of polygons closed up to translation. It turns out that those invariants do not extend to single-valued functions on the space  $\mathcal P_n^{E}$ of polygons closed up to isometry. To get well-defined invariants, we consider a double covering space
$$
\tilde{\mathcal P}_n^{{E}} := \{ ((v_i), \alpha) \in \mathcal P_n^{E} \times S^1 \mid (v_i) \mbox{ has monodromy } z \mapsto \alpha^2 z + \beta \mbox{ for some } \beta \in \C \}.
$$
%%Let  $\mathrm{E}$ be the group of orientation-preserving isometries, and $\mathcal P_n^{\mathrm{E}}$ be the space of $n$-gons closed up to an orientation-preserving isometry.
%Consider now the space $\mathcal P_n^{\mathrm{E}}$ of polygons closed up to an orientation-preserving isometry. In that case, the recutting invariants are functions on the double covering $\tilde{\mathcal P}_n^{\mathrm{E}}$ of that space defined as
%$$
%\tilde{\mathcal P}_n^{\mathrm{E}} := \{ ((v_i), \alpha) \in \mathcal P_n^{\mathrm{E}} \times S^1 \mid (v_i) \mbox{ has monodromy } z \mapsto \alpha^2 z + \beta \mbox{ for some } \beta \in \C \}.
%$$
The projection map $\tilde{\mathcal P}_n^{{E}}  \to {\mathcal P}_n^{{E}} $ takes a pair $((v_i), \alpha)$ to $(v_i)$, so that elements of $\tilde{\mathcal P}_n^{{E}} $ can be thought of as polygons closed up  to  isometry with a chosen square root of the rotational part of the monodromy. Recuttings act on $\tilde{\mathcal P}_n^{{E}} $ by acting on the first component.
%%In other words, the elements $\tilde{\mathcal P}_n^{\mathrm{E}} $ are pairs (polygons with monodromy $z \mapsto \alpha z + \beta$)
%
%
% %$z \mapsto e^{\i \phi} z + \beta$. Fix a value of the square root of $e^{\i \phi}$, and denote it by $e^{{\i \phi}/{2}}$. 
Consider $((v_i), \alpha) \in \tilde{\mathcal P}_n^{{E}} $, and let $z_i = v_i - v_{i-1}$ be the edge vectors of the polygon $(v_i)$. Let
%%$$
%%F_i(t) := (1 + z_i {\j}t)(1 + z_{i+1} {\j}t)\dots(1 + z_{i+n - 1} {\j}t)
%%$$
\begin{equation}\label{eq:fpoly2}
f(t):= (1 + z_1 {\j}t)(1 + z_{i+1} {\j}t)\cdots(1 + z_{n} {\j}t)\alpha.
\end{equation}
%where $\alpha^{1/2}$ is an arbitrarily chosen square root of $\alpha$. Similarly to the above, one shows that the similarity class of $f$ descends to the quotient by the $\mathrm{E}$ action and is invariant under recutting. The formulas for invariants $I_{2k}$ are modified accordingly:
\begin{proposition}[cf. Proposition \ref{prop:Lax}]\label{prop:Lax2}
The similarity class of the polynomial $f(t)$ in the skew-field $\tilde \H[[t]]$ is invariant under both the action of the group $E$ of isometries, and the recutting group action.
\end{proposition}
%The proof of this proposition is a little more tricky than that of its counterpart for polygons closed up to translation (Proposition \ref{prop:Lax}). 
We begin with a lemma, which will also be useful later on. Define the \textit{gauge action} of $(\C^*)^n$ on $(\tilde \H[t])^n$ by
\begin{equation}\label{eq:ga}
(g_i \in \tilde \H[t])_{i=1}^n  \mapsto (\lambda_i g_i \lambda_{i+1}^{-1})_{i=1}^n 
\end{equation}
where $\lambda_i \in \C^*$, and the indices are understood cyclically, i.e. the index $n+1$ is equivalent to the index $1$. Clearly, if two $n$-tuples $g_i(t) \in \tilde \H[t]$ and $\tilde g_i(t) \in \tilde \H[t]$ are gauge-equivalent, then the products $g_1(t)  \cdots  g_n(t)$ and $\tilde g_1(t)  \cdots  \tilde g_n(t)$ are similar. 

\begin{lemma}\label{lemma:gauge}
Let $((v_i), \alpha) \in \tilde{\mathcal P}_n^{{E}} $, and let $z_i = v_i - v_{i-1}$ be the edge vectors of the polygon $(v_i)$. 
Then the $n$-tuples 
$$
g_1 := 1 + z_{i+1} {\j}t, \quad \dots \quad g_{n-2} := 1 + z_{i+n-2} {\j}t, \quad g_{n-1} := (1 + z_{i+n-1} {\j}t)\alpha,  \quad g_n := 1 + z_{i} {\j}t$$ and $$\tilde g_1 := 1 + z_{i+1} {\j}t, \quad \dots\quad \tilde g_{n-1}:= 1 + z_{i+n-1} {\j}t, \quad \tilde g_n:= (1 + z_{i+n} {\j}t)\alpha
$$
are gauge-equivalent.
%(1 + z_i {\j}t)(1 + z_{i+1} {\j}t)\dots(1 + z_{n} {\j}t)\alpha
\end{lemma}
\begin{proof}
Take $\lambda_1 = \dots = \lambda_{n-1} = 1$, $\lambda_n = \alpha$.  Then one clearly has $\lambda_j g_j \lambda_{j+1}^{-1} = \tilde g_j$ for  $j = 1, \dots n-1$. 
%%Then, for $j = 1, \dots n-2$, we clearly have
%$\lambda_i(1 + z_{i+1} {\j}t)\lambda_{i+1}^{-1}= 1 + z_{i+j} {\j}t.$
%Furthermore,
%$\lambda_{n-1}((1 + z_{i+n-1} {\j}t)\alpha)\lambda_{n}^{-1}= 1 + z_{i+n-1} {\j}t.$$
Furthermore,
$$
\lambda_ng_n\lambda_1^{-1} = \alpha + z_{i}\alpha {\j}t =  \alpha + z_{i+n}\alpha^{-1} {\j}t = \alpha + z_{i+n}\bar \alpha {\j}t = (1 + z_{i+n} {\j}t)\alpha = \tilde g_n,
$$
where on the second step we used that $z_{i+n} = \alpha^2 z_i$, on the third step we used that $\alpha \in S^1$ and hence $\alpha^{-1} = \bar \alpha$, and on the second last step we used that $\bar \alpha \j = \j \alpha$. So we see that $\lambda_j g_j \lambda_{j+1}^{-1} = \tilde g_j$ for all $j$, as needed.
\end{proof}

\begin{proof}[Proof of Proposition {\ref{prop:Lax2}.}]
Let $
f_i(t) := (1 + z_i{\j}t)\cdots(1 + z_{i+n - 1} {\j}t)\alpha$. Then $f_i(t)$ is similar to the polynomial $\tilde f_i(t) := (1 + z_{i+1}{\j}t)\cdots(1 + z_{i+n - 1} {\j}t)\alpha (1 + z_i{\j}t)$, which, by Lemma \ref{lemma:gauge}, is similar to $f_{i+1}(t)$. So all $f_i$ are similar to each other and in particular to $f_1 = f$. The rest of the proof is the same as for Proposition \ref{prop:Lax}.
\end{proof}

\begin{proposition}
For a polygon closed up to isometry, one has
\begin{equation}\label{eq:eucinv0}
\begin{gathered}
\bar f(t) f(t)= 1 + \sum_k E_k t^{2k} ,\\
\mathrm{Re}\, f(t) = \sum_{k = 0}^{\lfloor n/2 \rfloor} (-1)^k I_{k}t^{2k},
\end{gathered}
\end{equation}
where
\begin{equation}\label{eq:eucinv}
\begin{gathered}
E_k :=\sum_{i_1 <  \dots < i_{k}} |z_{i_1}|^2 \dots |z_{i_k}|^2,  \\
I_{k} := \mathrm{Re} \left(\,\alpha\!\!\!\sum_{i_1 <  \dots < i_{2k}} z_{i_1}\bar z_{i_2} \dots z_{i_{2k - 1}} \bar z_{i_{2k}}\right).
\end{gathered}
\end{equation}
%where $k = 1, \dots, \lfloor n/2 \rfloor$.
\end{proposition}
\begin{proof}
See the proof of Proposition \ref{prop:inv}.
\end{proof}
%\begin{remark}
%Note that $I_0 = \mathrm{Re}\, \alpha$. The preservation of this function by recutting is trivial, since by definition of the recutting action on $\tilde{\mathcal P}_n^{{E}}/E$ it does not change $\alpha$. 
%\end{remark}

So recutting action on $n$-gons closed up to isometry has has $\lfloor 3n/2\rfloor + 1$ invariants, namely $n$ elementary symmetric polynomials $E_1, \dots, E_n$ of squared side lengths $|z_i|^2$ (whose invariance is obvious from the geometric definition of recutting), and $\lfloor n/2 \rfloor$ additional invariants $I_0, \dots, I_{\lfloor n/2 \rfloor}$. Note that $I_0 = \mathrm{Re}\, \alpha$, which is trivially invariant since by definition of the recutting action on $\tilde{\mathcal P}_n^{{E}}/E$ it does not change $\alpha$. 
%Since the contruction of these invariants involves a choice of a square root of $\alpha$, they are defined up to sign. So one can either view them as functions defined on a double covering of the corresponding polygon space, or just consider their squares. We will not discuss the details of this construction in the present paper.
%\begin{remark}
%Since the functions $I_{2k}$ change sign when $\alpha$ is multiplied by $-1$, they are not well-defined on the space  $\mathcal P_n^{\mathrm{E}}$ of $n$-gons closed up to  an orientation-preserving isometry. However, they are well-defined up to sign (and one can e.g. take their squares to obtain single-valued invariants). Also note that the functions $I_{2k}$ do not change under simultaneous rotations of all $z_i$, so they descend to the space ${ \tilde{\mathcal P}_n^{\mathrm{E}}}/{ \mathrm{E}}$, where $\mathrm{E}$ is the group of orientation-preserving isometries.
%\end{remark}
\subsection{Poisson geometry of special quaternionic polynomials}\label{subsec:pg}

In this section we show that the algebra $\tilde \H[t]$ of special quaternionic polynomials admits a multiplicative Poisson structure with nice properties. This structure can be obtained by extending the algebra $\tilde \H[t]$ to the algebra of special Laurent series in $t$ and endowing the latter with an $r$-matrix of trigonometric type. Here we use a different approach based on representing special quaternionic polynomials  as difference operators and then using a known Poisson structure on such operators.
\par
%The skew-field  $\tilde \H[[t]]$ of special quaternionic power series is defined analogously to $\tilde \H[t]$: a quaternionic power series is special if and only if all its even coefficients are complex numbers, while all odd coefficients belong to the complementary subspace $\mathrm{span}_\R\langle \j, \k \rangle$.  Say that $f,g \in \tilde \H[t]$ are similar if there exists an invertible formal power series $h \in \tilde \H[[t]]$ such that $hfh^{-1} = g$. 

\begin{proposition}\label{prop:sppb}
There exists a Poisson structure on the algebra $\tilde \H[t]$ of special quaternionic polynomials  with the following properties:
\begin{enumerate}
\item It is multiplicative, i.e. the multiplication map $\tilde \H[t] \times \tilde \H[t] \to \tilde \H[t]$ is Poisson.
\item Fixed degree polynomials form a Poisson subspace. 
\item The Poisson structure vanishes on constant (i.e. degree $0$) polynomials.
\item On linear polynomials $a + b \j t$, where $a,b \in \C$, the Poisson structure has the form
\begin{equation}\label{eq:pblsqp}
\{a, b\} = -\textstyle\frac{1}{2}ab, \quad \{a,\bar b\} = \frac{1}{2}a \bar b, \quad \{\bar a, \bar b\} = -\textstyle\frac{1}{2}\bar a\bar b, \quad \{\bar a, b\} = \frac{1}{2}\bar a  b, \quad \{a, \bar a\} = 0, \quad \{b, \bar b\} = 0.
\end{equation}
\item Central functions on $\tilde \H[t]$ Poisson commute (we say that a function $\chi\colon\tilde \H[t] \to \R$ is central if $\chi(f) = \chi(g)$ for any similar $f,g \in \tilde \H[t]$).
\item The function $\tilde \H[t] \to \R$ mapping $f(t)$ to $|f(0)|$ is a Casimir.
%All other pairwise brackets between $a, \bar a, b, \bar b$ either vanish or are determined by these by skew-symmetry and reality conditions. 
\end{enumerate}
\end{proposition}
\begin{remark} One can prove that $|f(t)|$ is a Casimir for any real $t$. That is equivalent to saying that all coefficients of the companion polynomial are Casimirs. \end{remark}
To prove Proposition \ref{prop:sppb} we recall the definition of a difference operator. Let $\mathbb K$ be a field, and let $\mathbb K^\infty = (\xi_i \in \mathbb K)_{i \in \Z}$ be the vector space of bi-infinite sequences valued in $\mathbb K$. A degree $d$ \textit{difference operator} over $\mathbb K$ is a linear map $\D \colon \mathbb K^\infty \to \mathbb K^\infty$ of the form
$
\D = \sum_{i= 0}^d a_i \mathcal T^i,
$
where $\mathcal T \colon  \mathbb K^\infty \to \mathbb K^\infty$ is the left shift operator $(\mathcal T(\xi))_i := \xi_{i+1}$, while each $a_i$ is an element of $\mathbb K^\infty $ acting on $\mathbb K^\infty$ by term-wise multiplication. A difference operator $\D$ is \textit{$n$-periodic} if its coefficients $a_i$ are $n$-periodic sequences, i.e. $(a_i)_{j+n} = (a_i)_j$.

\begin{proposition}\label{prop:h0DO}
As a graded associative algebra over reals, $\tilde \H[t]$ is isomorphic to the algebra of $2$-periodic difference operators $\D$ with complex coefficients such that $\mathcal T \D  \mathcal T^{-1} = \bar \D$.
\end{proposition}
\begin{proof}[Proof of Proposition \ref{prop:h0DO}]
The $\R$-algebra $\tilde \H[t]$ is generated by complex numbers $z \in \C$ and the polynomial $\j t$. The $\R$-algebra  of $2$-periodic difference operators $\D$  satisfying $\mathcal T \D \mathcal T^{-1} = \bar \D$ is generated by $2$-periodic bi-infinite sequences of the form $a_z := (\dots, z, \bar z, \dots)$ and the operator $\mathcal T$. In terms of these generators, the isomorphism between these two algebras is given by
$
z \mapsto a_z$,  $\j t \mapsto \mathcal T$. 
One easily checks that the relations between the generators on both sides are the same.
\end{proof}
\begin{proof}[Proof of Proposition \ref{prop:sppb}]
Consider the algebra of $2$-periodic difference operators with real coefficients. By \cite[Proposition 3.9]{izosimov2018pentagram}, this algebra carries a multiplicative Poisson structure such that the map $\D \mapsto \mathcal  T\D \mathcal T^{-1}$ is a Poisson automorphism. From the latter and Lemma \ref{lemma:cr} it follows that the extension of this structure to operators with complex coefficients restricts to operators such that $\mathcal T \D  \mathcal T^{-1} = \bar \D$. This gives a multiplicative bracket on $\tilde \H[t]$. The desired properties of that bracket follow from properties of the bracket on difference operators. Namely, properties 1, 2, 3, 5 follow from the corresponding parts of \cite[Proposition 3.9]{izosimov2018pentagram}, property~4 follows from \cite[equation (14)]{izosimov2018pentagram}, while property 6 follows from  \cite[Proposition 3.19]{izosimov2018pentagram} combined with the fact that the determinant is a Casimir of the standard Poisson structure on $GL_n$.
\end{proof}
\begin{remark}
We changed the sign of the bracket from \cite{izosimov2018pentagram} for conformance with our cluster bracket \eqref{eq:cpb}.
\end{remark}

\subsection{A recutting-invariant Poisson structure}\label{subsec:ps}
In this section we describe a Poisson bracket on the double cover $\tilde{\mathcal P}_n^{{E}}/E$ of the space ${\mathcal P}_n^{{E}}/E$ of polygons closed up to isometry and considered modulo isometries.  This bracket is preserved by the recutting and has a property that the invariants defined in the Section \ref{sec:inveuc} Poisson commute. Furthermore, this bracket descends to the space ${\mathcal P}_n^{{E}}/E$ and is taken by the natural map ${\mathcal P}_n^{{E}}/E \to { \mathcal P_n^{S}}/{ S}$ to the cluster bracket~\eqref{eq:cpb}.
\par

Let $\tilde \H^*[t] = \{ g(t) \in \tilde \H[t] \mid g(0) \neq 0 \}$ be the space  of special quaternionic polynomials with non-vanishing free term. Also, let   $\tilde \H_k^*[t] $ be its subset consisting of polynomials of degree strictly equal to $k$. Then, by part~2 of Proposition~\ref{prop:sppb},  the space $\tilde \H_k^*[t] $ is a Poisson submanifold of $\tilde \H[t] $. Let $ d = (d_1, \dots, d_n) \in \Z_+^n$. Then $Z_{ d} := \tilde \H_{d_1}^*[t] \times \dots \times \tilde \H_{d_n}^*[t]$ carries a product Poisson structure. Furthermore, by part 6 of Proposition~\ref{prop:sppb}, that Poisson structure restricts to $X_{ d} := \{(g_1(t), \dots, g_n(t)) \in  X_{ d} \mid |g_1(0) \cdots g_n(0)| = 1\}$. Also note that since the Poisson structure on special quaternionic polynomials vanishes on $\tilde \H_0^*[t] = \C^* $, the gauge action \eqref{eq:ga} of $(\C^*)^n$ on $X_d$ is Poisson, so the Poisson structure descends to $X_d / (\C^*)^n$.

We now show that for $d = (1, \dots, 1) \in \Z_+^n$,  the space $X_{d} / (\C^*)^n$ can be identified with $\tilde{\mathcal P}_n^{{E}}/E$. The following is straightforward:
\begin{proposition}\label{prop:gaugeQ0}
For any $d \in \Z_+^n$, every orbit of the gauge action of $(\C^*)^n$ on $X_d$ has a representative $(g_i(t))_{i=1}^n$ with $g_1(0) = \dots = g_{n-1}(0) = 1$ and $|g_n(0)| = 1$. Such a representative is unique up to a transformation of the form $(g_i(t)) \mapsto (\lambda g_i(t) \lambda^{-1}) $ where $\lambda \in S^1$.
\end{proposition}

%Consider the Poisson structure on $\tilde \H[t]$ is the one described in Proposition{prop:sppb}.
%We now return to the problem of constructing a recutting-invariant Poisson structure on the space  ${ {\mathcal P}_n^{\C}}/{ \mathrm{E}}$. 
%Consider the space $X \subset \tilde \H[t]^n$ which consists of $n$-tuples of the form $(g_1, \dots, g_n)$, where $g_i = a_i + b_i \j t$, $a_i, b_i \in \C^*$, and let $X_1$ be the subset of elements of $X$ such that $a_1 \cdot \ldots \cdot a_n = 1$. Define the \textit{gauge} action of $(\C^*)^{n}$ on $X$ by
%\begin{equation}\label{eq:ga}
%(g_1, g_2,\dots, g_n) \mapsto (\lambda_1 g_1 \lambda_{2}^{-1},\lambda_2 g_1 \lambda_{3}^{-1} \dots, \lambda_n g_n \lambda_{1}^{-1}).
%\end{equation}
%This action preserves the subset $X_1 \subset X$. The following is straightforward:
\begin{corollary}\label{prop:gaugeQ}
Every orbit of the gauge action of $(\C^*)^{n}$ on $X_{1, \dots, 1}$ has a representative of the form \begin{equation}\label{eq:gaugeSection}(1 + z_1 \j t, \dots, (1 + z_n \j t)\alpha),\end{equation} where $z_i \in \C^*, \alpha \in S^1$. Such a representative is unique up to simultaneous rotation of all $z_i$.
\end{corollary}
It follows that the quotient $X_{1, \dots, 1} / (\C^*)^{n}$ can be identified with the space of $\tilde{\mathcal P}_n^{{E}}/E$. Namely, given a gauge equivalence class in $X_{1, \dots, 1} / (\C^*)^{n}$, one finds its representative of the form \eqref{eq:gaugeSection} and then maps it to a pair $((v_i), \alpha)$, where $v_i$ is a polygon whose sequence of edge vectors $z_i$ is obtained from numbers entering~\eqref{eq:gaugeSection} by imposing quasi-periodicity condition $z_{i+n} = \alpha^2 z_i$. Since the canonical form \eqref{eq:gaugeSection} is defined up to simultaneous rotation of all $z_i$, this gives a well-defined map $X_{1, \dots, 1} / (\C^*)^{n} \to \tilde{\mathcal P}_n^{{E}}/E$. Moreover, it follows from Proposition~\ref{prop:zphicoords} that the so-defined map is a bijection. In particular, we get a Poisson bracket on the space $ \tilde{\mathcal P}_n^{{E}}/E$. The following proposition summarizes its properties.
\begin{proposition}\label{prop:pbpoly}
\begin{enumerate}\item 
The Poisson bracket on $ \tilde{\mathcal P}_n^{{E}}/E$ is invariant under recutting.
%\item The functions $|z_i|^2$ are its Casimirs.
\item The invariants $E_1, \dots, E_n, I_{0}, \dots, I_{\lfloor n/2 \rfloor}$ Poisson commute.
\item Moreover, $E_1, \dots, E_n, I_{0}$ are Casimirs. %If $n$ is even, then $I_{n/2}$ is also a Casimir. Connected components of joint level sets of those $2\lfloor n/2 \rfloor + 2$ Casimirs coincide with symplectic leaves.
\item The bracket descends to the space  $ {\mathcal P}_n^{{E}}/E$.
\item In coordinates $|z_i|, \phi_i$, the bracket on $ {\mathcal P}_n^{{E}}/E$ has the following form: $|z_i|$ are Casimirs, while the bracket of $\phi_i$'s is given by \eqref{eq:phibr}. The function $\sum \phi_i$ is also a Casimir. For even $n$ there is, in addition, a Casimir $\phi_1 + \phi_3 + \dots + \phi_{n-1}$. Joint level sets of $|z_i|$ and these Casimirs are symplectic leaves.
%\begin{equation}\{\phi_i, \phi_j\} = \delta_{i+1, j} - \delta_{i-1, j}.\end{equation}
%\item When $n$ is odd, symplectic leaves of the bracket on $ {\mathcal P}_n^{{E}}/E$ are joint levels sets of the functions $|z_1|, \dots, |z_n|$ and $\phi_1 + \dots + \phi_n$. When $n$ is even, symplectic leaves are joint levels sets of  $|z_1|, \dots, |z_n|$, $\phi_1 + \dots + \phi_n$, and $\phi_1 + \phi_3 + \dots + \phi_{n-1}$.
\item The map ${\mathcal P}_n^{{E}}/E \to { \mathcal P_n^{S}}/{ S}$ takes the bracket on ${\mathcal P}_n^{{E}}/E$ to the cluster bracket~\eqref{eq:cpb}.
\end{enumerate}
\end{proposition}

We begin with a lemma. Consider a relabeling of vertices map on ${\mathcal P}_n^{{E}}$ given by $(v_i) \mapsto (\tilde v_i)$ where $\tilde v_i := v_{i+1}$. This map induces an order $n$ map $\mathcal S$ on ${\mathcal P}_n^{{E}}/E$ and $ \tilde{\mathcal P}_n^{{E}}/E$. %Consider also a map $\hat{\mathcal T}$ on $X_{1, \dots, 1} / (\C^*)^{n}$ induced by the cyclic shift map $(g_1, \dots, g_n) \mapsto (g_2, \dots, g_n, g_1)$.

\begin{lemma}\label{lemma:cyclic}
The Poisson bracket on  $ \tilde{\mathcal P}_n^{{E}}/E$ is invariant under the map  $\mathcal S$.
%The identification between the spaces $X_{1, \dots, 1} / (\C^*)^{n}$ and $\tilde{\mathcal P}_n^{{E}}/E$ intertwines .
\end{lemma}
\begin{proof}
Consider the map $\hat{\mathcal S}$ on $X_{1, \dots, 1} / (\C^*)^{n}$ induced by the cyclic shift map $(g_1, \dots, g_n) \mapsto (g_2, \dots, g_n, g_1)$. Clearly, this map is Poisson. So it suffices to show that the identification between the spaces $X_{1, \dots, 1} / (\C^*)^{n}$ and $\tilde{\mathcal P}_n^{{E}}/E$ intertwines the maps $\hat{\mathcal S}$ and $\mathcal S$. This amounts to saying that the cyclic shift of \eqref{eq:gaugeSection} to the left is gauge equivalent to $(1 + z_2 \j t, \dots, (1 + z_{n+1} \j t)\alpha)$. But this is exactly the content  of Lemma \ref{lemma:gauge} for $i = 1$.
\end{proof}

\begin{proof}[Proof of Proposition \ref{prop:pbpoly}] We prove the parts of the proposition in a convenient order. First we prove part~1. 
By Lemma \ref{lemma:cyclic}, the bracket on  $ \tilde{\mathcal P}_n^{{E}}/E$ is invariant under relabeling of vertices, so it suffices demonstrate the invariance of that bracket under recutting $\rho_{1}$. %Let  $Y \subset \tilde \H[t]^n$ be the Poisson  subspace which consists of $(n-1)$-tuples of the form $(g_0, g_3, \dots, g_{n})$, where $g_0$ is quadratic, all other $g_i$ are linear, and all $g_i$ have non-zero free terms. 
Let $d:= (1, \dots, 1) \in \Z_+^n$ and $ d' := (2,1,\dots, 1) \in \Z^{n-1}$. Consider the map $\mathcal M \colon  X_{d} \to X_{d'}$ given by $(g_1, \dots, g_{n}) \mapsto (g_1g_2, g_3, \dots, g_n)$. It is a Poisson map due to multiplicativity of the Poisson structure on $\tilde \H[t]$. %Consider also the gauge action of $(\C^*)^{n-1}$ on $Y$. It is also Poisson. 
So it descends to a Poisson map $\mathcal M'$ between gauge quotients $X_d / (\C^*)^{n} \to X_{d'} / (\C^*)^{n-1}$. %Denote by $Y_1$ the image of $X_1$ under the map $\mathcal M$. Then $\mathcal M'(X_1 / (\C^*)^{n} ) = Y_1 / (\C^*)^{n-1}$. 
Using Corollary \ref{prop:gaugeQ}, identify elements of $X_d / (\C^*)^{n}$ with $n$-tuples of the form~\eqref{eq:gaugeSection} where $z_n$ is positive real. Using Proposition \ref{prop:gaugeQ0}, identify $X_{d'} / (\C^*)^{n-1}$ with the subspace of $X_{d'}$ which consists of  $(g_1, \dots, g_{n-1}) \in X_{d'}$ such that $g_i(0) = 1$ for all $i$, and $g_{n-1} = (1 + \beta \j t)\alpha$ with $\beta$ positive real and $\alpha \in S^1$. With these identifications, the map $\mathcal M' \colon X_d / (\C^*)^{n} \to X_{d'} / (\C^*)^{n-1}$ is given by
\begin{equation}\label{eq:cov}(1 + z_1 \j t, \dots, (1 + z_n \j t)\alpha) \mapsto ((1 + z_1 \j t)(1 + z_2 \j t),1 + z_3 \j t, \dots, (1 + z_n \j t)\alpha).\end{equation} 
According to Proposition \ref{prop:rqp} and Remark \ref{rem:rqp}, the mapping \eqref{eq:cov} is generically a $2$-to-$1$ covering whose only non-trivial deck transformation is given by the recutting $\rho_1$.  So $\rho_1$ is Poisson as a deck transformation of a Poisson covering. Thus, part 1 is proved. \par
%To prove the second statement, consider the functions $ |b_i|^2/|a_i|^2$ on $X$. As follows from  formulas \eqref{eq:pblsqp}, both $|a_i|^2$ and $|b_i|^2$ are Casimirs, so their ratios are Casimirs too.
%  Furthermore, those ratios are invariant under the gauge action and thus descend to the quotient $X / (\C^*)^{n}$. Moreover, the restriction of so-obtained functions on $X / (\C^*)^{n}$  to $X_1 / (\C^*)^{n}$ are precisely $|z_i|^2$. %(Indeed, since we identify the quotient $X_1 / (\C^*)^{n}$ with the section~\eqref{eq:gaugeSection} of the projection $\pi \colon X_1 \to X_1 / (\C^*)^{n}$, evaluating a function on  $X_1 / (\C^*)^{n}$ is equivalent to restricting the pullback of that function by $\pi$ to  $n$-tuples of the form \eqref{eq:gaugeSection}.)
 % So $|z_i|^2$ are restrictions of functions whose pullbacks are Casimirs, and thus are themselves Casimirs.
  
  We now prove part 2. Recall that the invariants $E_1, \dots, E_n, I_{0}, \dots, I_{\lfloor n/2 \rfloor}$ on the space $X_{1, \dots, 1} / (\C^*)^{n} \simeq \tilde{\mathcal P}_n^{{E}}/E$ are defined as central functions on $\H_0[t]$ applied to the polynomial \eqref{eq:fpoly2}. Therefore, the pullbacks of those invariants to $X_{1, \dots, 1}$  %central functions applied to 
 coincide with pullbacks of central functions  $\tilde \H[t]$ by the product map   $(g_1, \dots, g_n) \in X_{1, \dots, 1}  \mapsto g_1 \cdots  g_n \in \tilde \H[t]$. Since the product map is Poisson, it follows that the invariants Poisson commute, as desired.\par
Next we prove part 4. Observe that the projection $\tilde{\mathcal P}_n^{{E}}/E \to {\mathcal P}_n^{{E}}/E$ can be viewed as quotient by the involution $ \alpha \mapsto -\alpha$. That involution lifts to a Poisson involution of  $X_{1, \dots, 1}$ given by $(g_1, \dots, g_n) \mapsto (g_1, \dots, -g_n)$ and is, therefore, Poisson as well. So the space  ${\mathcal P}_n^{{E}}/E$ inherits a Poisson structure from  $\tilde{\mathcal P}_n^{{E}}/E$ as a quotient by Poisson involution. 
 
 Now prove parts 5 and 6. Consider functions $\tilde y_1, \dots, \tilde y_n$ on $X_{1, \dots, 1}$ defined by
 $$
 \tilde y_i(a_1 + b_1 \j t, \dots, a_n + b_n \j t) := \frac{a_i b_{i+1}}{b_i \bar a_{i+1}}
 $$
 where the index $i$ is understood modulo $n$. It is easy too see that $\tilde y_i$ are gauge-invariant and thus descend to the quotient $X_{1, \dots, 1} / (\C^*)^{n} \simeq \tilde{\mathcal P}_n^{{E}}/E$. To compute their pushforward to the quotient one just needs to restrict them to $X_{1, \dots, 1}$ elements of the form \eqref{eq:gaugeSection}. By doing so one finds that the pushforward of $\tilde y_i$ is the ratio $y_i$ of consecutive edge vectors. That allows one to find Poisson brackets of $y_i$ by computing brackets of $\tilde y_i$ and then pushing the result forward to the quotient. The brackets of $\tilde y_i$ can be found using formulas \eqref{eq:pblsqp}. As a result, one finds that the brackets of $y_i$'s are given by \eqref{eq:cpb}, which precisely means that the map ${\mathcal P}_n^{{E}}/E \to { \mathcal P_n^{S}}/{ S}$ (equivalently, $\tilde{\mathcal P}_n^{{E}}/E \to { \mathcal P_n^{S}}/{ S}$) is Poisson. Thus part 6 is proved. In view of Remark \ref{rem:phibr}, this also proves that the brackets of $\phi_i$ are of the form \eqref{eq:phibr}. So to complete the proof of part 5 it suffices to show that $|z_i|$ are Casimirs (the description of additional Casimirs and symplectic leaves easily follows using the constant form of the bracket~\eqref{eq:phibr}). In the same way as we showed that the pullbacks of $y_i$ to the quotient are given by $\tilde y_i$, one shows that the pull-backs of $|z_i|$ are 
$$
 \zeta_i(a_1 + b_1 \j t, \dots, a_n + b_n \j t) := \frac{|b_i|}{|a_i|}.
 $$
 These are easily seen to be Casimirs using formulas \eqref{eq:pblsqp}. Thus, part 5 is proved. %Part 6 now follows from the constant form of the bracket . 
 
 Finally, we prove part 3. The functions $E_i$ are clearly Casimirs, since they are symmetric functions of the Casimirs $|z_i|^2$. Also notice that function $\alpha$ on  $\tilde{\mathcal P}_n^{{E}}/E$ is a  (complex-valued) Casimir, because the pushforward of $\alpha^2$ to ${\mathcal P}_n^{{E}}/E$ is the function $\exp(\i (\phi_1 + \dots + \phi_n))$, which is a Casimir. Therefore, $I_0 = \mathrm{Re}\,\alpha$ is also a Casimir. Finally, for even $n $, the function $I_{n/2} = \sqrt{E_n}\mathrm{Re}(\alpha \exp(\i (\phi_1 + \phi_3 + \dots + \phi_{n-1})) )$ is a Casimir since so are $E_n$, $\alpha$, and $\phi_1 + \phi_3 + \dots + \phi_{n-1}$. %As for the description of the symplectic leaves, it easily follows from the constant form of the bracket \eqref{eq:phibr}. 
 Thus, the proposition is proved.
   \end{proof}
   \begin{remark}[cf. Remark \ref{rem:manyCas}]
The submanifold of ${\mathcal P}_n^{{E}}/E$ defined by $\sum \phi_i = 0 \mod 2\pi$ is the space ${\mathcal P}_n^{{T}}/E$ of isometry classes of polygons closed up to translation. So, since $\sum \phi_i$ is a Casimir, it follows that ${\mathcal P}_n^{{T}}/E$ is a Poisson submanifold of ${\mathcal P}_n^{{E}}/E$. As for the space ${\mathcal P}_n/E$ os isometry classes of closed polygons, it is not a Poisson submanifold of  ${\mathcal P}_n^{{E}}/E$ because  ${\mathcal P}_n/S$ is not a Poisson submanifold of ${\mathcal P}_n^{{S}}/S$.
\end{remark}
   
%\begin{remark}
%In the same way one can show the Poisson bracket on ${ {\mathcal P}_n^{\C}}/{ \C}$ constructed in the proof is invariant under recuttings $\rho_2, \dots, \rho_{n-1}$. However, it is not invariant under recutting $\rho_n$.
%\end{remark}
%
%%Consider the mapping $\psi \colon X \to \tilde \H[t]^{n-1}$ given by
%%$$
%%(1 + z_1 \j t, \dots, 1 + z_n \j t) \mapsto ((1 + z_1 \j t)(1 + z_2 \j t),1 + z_3 \j t, \dots, 1 + z_n \j t).
%%$$
%\end{proof}
%We now identify the gauge quotient $X / (\C^*)^{n+1}$ with the space ${ \tilde{\mathcal P}_n^{\mathrm{E}}}/{ \mathrm{E}}$. The following is straightforward:
%\begin{proposition}\label{prop:gaugeQ}
%Every orbit of the gauge action of $(\C^*)^{n+1}$ on $X$ has a representative of the form $(1 + z_1 \j t, \dots, 1 + z_n \j t, \alpha)$, where $z_i \in \C^*, \alpha \in \C, |\alpha| = 1$. Such a representative is unique up to simultaneous rotation of all $z_i$.
%\end{proposition}
%Now note that the space ${ \tilde{\mathcal P}_n^{\mathrm{E}}}/{ \mathrm{E}}$ is parametrized by the side vectors $z_1, \dots, z_n \in \C^*$, considered up to simultaneous rotation, and the value $\alpha \in S^1$, the square root of the rotational part of the monodromy. So, Proposition~\ref{prop:gaugeQ} provides an identification between the spaces $X / (\C^*)^{n+1}$ and ${ \tilde{\mathcal P}_n^{\mathrm{E}}}/{ \mathrm{E}}$. So, in view of Proposition~\ref{prop:gap}, we get a Poisson bracket on the space ${ \tilde{\mathcal P}_n^{\mathrm{E}}}/{ \mathrm{E}}$.

\subsection{Integrability}\label{subsec:int}
In this section we prove Theorem \ref{thm1} on integrability of recutting on the space $\mathcal P_n^E / E$ of isometry classes of polygons closed up to isometry. Given the result of Proposition \ref{prop:pbpoly}, it suffices to show that the integrals $E_1, \dots, E_n$, $I_0, \dots, I_{\lfloor n/2 \rfloor}$ are functionally independent. This is provided by the following:
\begin{lemma}\label{lemma:ind1}
The mapping $(\C^*)^n \times S^1 \to \R^{\lfloor 3n/2 \rfloor + 1}$ taking $(z_1, \dots, z_n, \alpha) \in (\C^*)^n \times S^1$ to the values of $E_1, \dots, E_n, I_0, \dots, I_{\lfloor n/2 \rfloor}$ defined by \eqref{eq:eucinv} is a submersion away from a set of measure zero. 
\end{lemma}
\begin{proof}
The cases $n = 1, 2$ are straightforward so from this point on assume $n \geq 3$. 
Let $\mathcal A:= (\C^*)^n \times S^1$ and let $\mathcal C$ be the set of pairs of real polynomials $g ,h$ with the following properties: both are even, the degree of $g$ is exactly $2n$, the degree of $h$ is at most $2[n/2]$, and $g(0) = 1$.
%Denote by $\tilde \R_d[t]$ the space of degree $d$ even polynomials in $s$ with real coefficients, and let $\tilde \R_d^1[t] := \{ h \in \tilde \R_d[t] \mid h(0) = 1\}$. 
Then the desired statement can be reformulated as follows. Consider the map $\Psi \colon \mathcal A \to \mathcal C $ that takes  $(z_1, \dots, z_n, \alpha) \in (\C^*)^n \times S^1$ to the polynomials \eqref{eq:eucinv0}. Then $\Psi$  is a submersion away from a set of measure zero. We prove this result by representing $\Psi$ as a composition of two maps. Let $\mathcal B$ be the set of special quaternionic polynomials $f \in \tilde \H[t]$ of degree strictly $n$ and such such that $|f(0)| = 1$. Then we have a map  $\Psi_1 \colon \mathcal A \to \mathcal B$ that sends  $(z_1, \dots, z_n, \alpha) \in \mathcal A$ to the special quaternionic polynomial $f(t)$ given by \eqref{eq:fpoly2}. Further, we have a map $\Psi_2 \colon \mathcal B \to \mathcal C  $ given by $f(t) \mapsto \bar f(t) f(t), \mathrm{Re}\,f(t)$. Clearly $\Psi = \Psi_2 \circ \Psi_1$. So it suffices to show that both maps $\Psi_1, \Psi_2$ are submersions away from a set of measure zero. Since both maps are polynomial, that is equivalent to saying that the images of both maps have non-empty interior. First consider $\Psi_1$. Its image consists of those polynomials $f \in \mathcal B $ which can be factored into linear factors. According to Proposition~\ref{prop:facspec}, that happens if and only if all complex roots of the companion polynomials of $f$ are on the imaginary axis. Since the companion polynomial of a special quaternionic polynomial is necessarily even, the set of $f$ with the desired property has non-empty interior, as needed. \par
Now consider $\Psi_2$. First assume that $n$ is odd, $n = 2k+1$. Consider $(g, h) \in \mathcal C$. By definition of $\mathcal C$, both $g,h$ are even, have degrees at most $4k+2$ and $2k$ respectively (with degree of $g$ being exactly $4k + 2$), and $g(0) = 1$. Therefore, the function $g(t) - h(t)^2 - 1 + h(0)^2$ is an even polynomial of degree at most $4k + 2$ vanishing at the origin. So, there exists a polynomial $\zeta(s)$ of degree at most $2k$  such that
\begin{equation}\label{eq:zetapoly}
g(t) - h(t)^2 - 1 + h(0)^2 = t^2\zeta(t^2).
\end{equation}
Say that $\zeta(s)$ is \textit{strictly positive} if its coefficient of $s^{2k}$ is positive, and $\zeta(s) > 0$ for any $s \in \R$. The set of strictly positive polynomials is open in the space of all real polynomials of degree at most $2k$. Now consider the subset $\Sigma$ of $\mathcal C$ consisting of pairs $g(t) ,h(t)$ such that $|h(0)| < 1$ and the associated polynomial $\zeta(s)$ defined by~\eqref{eq:zetapoly} is strictly positive. That is an open subset of $\mathcal C$. Furthermore, $(t^{2n} + t^2 + 1, 0)\in \Sigma$, so $\Sigma$ is non-empty. Let us show that $\Sigma$ is contained in the image of $\Psi_2$. Assume $(g,h) \in \Sigma$. Consider the polynomial  $\zeta(s)$ defined by~\eqref{eq:zetapoly}. Since $\zeta$ is positive of degree $2k$, there exist real polynomials $u, v$ of degree at most $k$ such that $u^2 + v^2 = \zeta$. Let also $a := \sqrt{1 - h(0)^2}$ (this is a real number since $|h(0)| < 1$). Consider a special quaternionic polynomial %Since $|h(0)| < 1$, there exists $\alpha \in \R$ such that $h(0) = \cos \alpha$. 
$
f(t) := h(t) + \i a + \j t u (t^2) + \k t v (t^2).
$
Then $\mathrm{Re}\, f(t) = h(t)$, and
$
\bar f(t) f(t) = g(t).
$
%and let $f(t) := f_1(t) + f_2(t) \i + f_3(t) \j + f_4(t) \k$. 
The latter in particular means that the degree of $f$ is exactly $n$. Furthermore, $f(0) = h(0) +  \i a = h(0) + \i \sqrt{1 - h(0)^2} $ so $f \in  \mathcal B $. 
Thus, $f \in \mathcal B$ and $\Psi_2(f) = (g,h)$, as needed.

Now assume $n$ is even, $n = 2k$. Consider  $(g, h) \in \mathcal C$. Then the polynomial $\zeta(s)$ defined by~\eqref{eq:zetapoly} has degree at most $ 2k - 1$. Consider the subset $\Sigma$ of $\mathcal C$ consisting of pairs $g(t) ,h(t)$ such that $|h(0)| < 1$, the associated polynomial $\zeta(s)$  has positive coefficient $b^2$ of $s^{2k-1}$, and the polynomial \begin{equation}\label{eq:xipoly}\xi(s) := \zeta(s) - b^2 s^{2k-1} - 2abs^{k-1} \end{equation} is strictly positive of degree $2k - 2$ (here, as before, we define $a := \sqrt{1 - h(0)^2}$). The set $\Sigma$ is open in $ \mathcal C$, and non-empty since  $( t^{2n} +t^{2n-2} + 2t^2 + 1,  0)\in \Sigma$. Let us show that $\Sigma$ is contained in the image of $\Psi_2$. Let $(g,h) \in \Sigma$. Consider the polynomial $\xi(s)$ defined by~\eqref{eq:xipoly}. Since $\xi$ is positive of degree $2k -2$, there exist  real polynomials $u,v$ of degree at most $k-1$ such that $u^2 +v^2 = \xi$. Let
$
f(t) := h(t) +\i(a + bt^{2k}) + \j t u(t^2) + \k t v (t^2).$ Then $f(t) \in   \mathcal B$, and $\Psi_2(f) = (g,h)$. Thus, the lemma is proved. \end{proof}
\begin{proof}[Proof of Theorem \ref{thm1}] We begin with part 1 of the theorem. Consider $E_1, \dots, E_n$, $I_0, \dots, I_{\lfloor n/2 \rfloor}$ as functions on the double covering of the space $ {\mathcal P}_n^E / E$. The functions~$E_j$ are symmetric polynomials of squared lengths of sides and hence descend to $ {\mathcal P}_n^E / E$. The functions $I_j$ are defined on $ {\mathcal P}_n^E / E$ up to sign. So $E_1, \dots, E_n$, $I_0^2, \dots, I_{\lfloor n/2 \rfloor}^2$ are well-defined functions on $ {\mathcal P}_n^E / E$ preserved by recutting. They Poisson commute by Proposition~\ref{prop:pbpoly} and are independent by Lemma \ref{lemma:ind1}. The number of those functions is $\lfloor 3n/2 \rfloor + 1$. The functions $E_1, \dots, E_n$, $I_0^2$ are Casimirs. For even $n$, the function $I_{n/2}^2$ is a Casimir too. So the total number $2 \lfloor n/2 \rfloor + 2$ of Casimirs coincides with the codimension of symplectic leaves. Therefore, the remaining first integrals $I_1^2, \dots, I^2_{\lceil n/2 \rceil - 1}$ are independent on generic symplectic leaves. Their number is exactly one half of the dimension of the leaves. So, the recutting dynamics on $ {\mathcal P}_n^E / E$ is indeed Arnold-Liouville integrable.

We now prove the second part. Since the squared lengths of sides considered as functions on $ {\mathcal P}_n^E / E$ are Casimirs, it follows that the symplectic leaves of  $ {\mathcal P}_n^E / E$ are compact. Therefore, by Arnold-Liouville theorem, connected components of generic joint level sets of first integrals are tori. Furthermore, since $E_1, \dots, E_n$ are symmetric functions of squared side lengths, their joint level sets are compact too. So any joint level set  of the recutting invariants $E_1, \dots, E_n$, $I_0^2, \dots, I_{\lfloor n/2 \rfloor}^2$ is compact and hence has finitely many connected components. For that reason, the subgroup of the recutting group preserving a given component must have finite index. That subgroup acts by translations by the discrete version of Arnold-Liouville theorem \cite{Veselov}. Thus, Theorem \ref{thm1} is proved. \end{proof}
\section{Recutting of closed polygons: non-Hamiltonian integrability}\label{sec:closed}
In this section we prove Theorem \ref{thm2} on integrability of recutting on the space $\mathcal P_n / E$ of isometry classes of closed polygons. Here is a version of Lemma \ref{lemma:ind1} adapted to closed polygons:
\begin{lemma}\label{lemma:ind2}
Assume that $n \geq 3$. Then the mapping $\{ (z_1, \dots, z_n) \in (\C^*)^n \mid \sum z_i = 0\} \to \R^{\lfloor 3n/2 \rfloor - 1}$ taking $(z_1, \dots, z_n)$ to the values of $E_1, \dots, E_n, I_2, \dots, I_{\lfloor n/2 \rfloor}$ defined by \eqref{eq:trinv} is a submersion away from a set of measure zero. 
\end{lemma}
\begin{proof}
The cases $n = 3, 4$ are straightforward so from this point on assume $n \geq 5$. 
The proof is a modification of that of Lemma \ref{lemma:ind1}. The sets $\mathcal A ,\mathcal B, \mathcal C$ and the maps $\Psi_1, \Psi_2$ are now defined as follows. The set $\mathcal A$ is $\{ (z_1, \dots, z_n) \in (\C^*)^n \mid \sum z_i = 0\}$. The set $\mathcal B$ consists of special quaternionic polynomials $f \in \tilde \H[t]$ of degree strictly $n$ and such such that $f(0) = 1$ and the coefficient of $t$ is equal to $0$. The map $\Psi_1 \colon \mathcal A \to \mathcal B$ takes  $(z_1, \dots, z_n) \in \mathcal A$ to the special quaternionic polynomial $f(t)$ given by \eqref{eq:fpoly}. The set $\mathcal C$ consists of pairs of real polynomials $g ,h$ of the form
$$
g(t) = 1 + \sum_{i=1}^n E_i t^{2i} ,\quad 
h(t) = 1 + \sum_{i = 1}^{\lfloor n/2 \rfloor} J_{i}t^{2i},
$$
where $E_1 = 2J_1$, and $E_n \neq 0$. The map $\Psi_2 \colon \mathcal B \to \mathcal C  $ is again given by $f(t) \mapsto \bar f(t) f(t), \mathrm{Re}\,f(t)$. The statement of the lemma is equivalent to saying that $\Psi_2 \circ \Psi_1$ is a submersion away from a set of measure zero. The proof that $\Psi_1$ has this property is exactly the same as in  Lemma \ref{lemma:ind1}. So it suffices to show that the image of $\Psi_2$ has non-empty interior.\par
 First assume that $n$ is odd, $n = 2k+1$. Consider $(g, h) \in \mathcal C$. Then  the polynomial $g(t)-h(t)^2$ is even, has degree at most $4k + 2$, and can be written as $(E_2 - J_1^2 - 2J_2 )t^4 + O(t^6)$. Therefore, there is a polynomial $\zeta(s)$ of degree at most $2k-2$ such that
\begin{equation}\label{eq:zetapoly2}
 g(t)-h(t)^2 - (E_2 - J_1^2 - 2J_2 )t^4 = t^6\zeta(t^2).
\end{equation}
 Consider the subset $\Sigma \subset \mathcal C$ consisting of pairs $g ,h$ such that $E_2 - J_1^2 - 2J_2 > 0$, and the polynomial $\zeta(s)$ is strictly positive. Then $\Sigma$ is open in $ \mathcal C$, and non-empty since $(t^{2n}+t^6+t^4 +1, 1) \in \Sigma$. Let us show that $\Sigma$ is contained in the image of $\Psi_2$. Let $(g,h) \in \Sigma$. Let also $a := \sqrt{E_2 - J_1^2 - 2J_2}$. %(Note that if $(g,h)$ is obtained from $(z_1, \dots, z_n) \in \mathcal A$ by means of the map $\Psi_2 \circ \Psi_1$, then by \eqref{eq:area} the number $A$ is the area of the polygon with edge vectors $z_i$.) 
 Since the polynomial $\zeta$ is strictly positive, there exist polynomials $u, v$ of degree at most $k-1$ such that $u^2 + v^2 = \zeta$. Let  %Since $|h(0)| < 1$, there exists $\alpha \in \R$ such that $h(0) = \cos \alpha$. 
$
f(t) := h(t) + \i at^2 + \j t^3 u (t^2) + \k t^3 v (t^2).
$
%and let $f(t) := f_1(t) + f_2(t) \i + f_3(t) \j + f_4(t) \k$. 
Then  $f \in  \mathcal B $, and $\Psi_2(f) = (g,h)$, as needed.

Now assume $n$ is even, $n = 2k$. Consider  $(g, h) \in \mathcal C$. Then the polynomial $\zeta(s)$ defined by~\eqref{eq:zetapoly2} has degree at most $2k - 3$ and its coefficient of $s^{2k - 3}$ is equal to $E_{2k} - J_k^2$. Consider the subset $\Sigma$ of $\mathcal C$ consisting of pairs $g(t) ,h(t)$ such that $E_2 - J_1^2 - 2J_2 = a^2 > 0$, $E_{2k} - J_k^2 = b^2 > 0$, and the polynomial \begin{equation}\label{eq:xipoly2}\xi(s) := \zeta(s) - b^2 s^{2k-3} - 2 ab s^{k-2}\end{equation} is strictly positive of degree $2k - 4$ (note that $k > 1$ so $\xi$ is indeed a polynomial). The set $\Sigma$ is open in $ \mathcal C$, and non-empty since $(t^{2n} + t^{2n-2}+2t^6 + t^4 +1,1)\in \Sigma$. Let us show that $\Sigma$ is contained in the image of $\Psi_2$. Let $(g,h) \in \Sigma$. Consider the polynomial $\xi(s)$ defined by~\eqref{eq:xipoly2}. Since $\xi$ is positive of degree $2k -4$, there exist  real polynomials $u,v$ of degree at most $k-2$ such that $u^2 +v^2 = \xi$. Let
$
f(t) := h(t) + \i (at^2 + bt^{2k} ) + \j t^3 u (t^2) + \k t^3 v (t^2).$ Then $f \in   \mathcal B$, and $\Psi_2(f) = (g,h)$. Thus, the lemma is proved. \end{proof}
\begin{proof}[Proof of Theorem \ref{thm2}]
We begin with the first part of the theorem. It is easy to see that  $\mathcal P_n / E$ is a codimension~$3$, and hence dimension $2n - 3$, submanifold of  $\mathcal P_n^E / E$.  By Lemma \ref{lemma:ind2}, the restrictions of the functions $E_1, \dots, E_n, I^2_2, \dots, I^2_{\lfloor n/2 \rfloor}$ on  $\mathcal P_n^E / E$ to $\mathcal P_n / E$ are functionally independent. Since on $\mathcal P_n / E$ we have $I_0^2 = 1$, this in particular means that the differentials of $I_0^2, E_1, \dots, E_n, I^2_2, \dots, I^2_{\lfloor n/2 \rfloor}$, considered as $1$-forms on the ambient manifold $\mathcal P_n^E / E$, are independent at generic points of $\mathcal P_n / E$. From this it follows that the commuting Hamiltonian vector fields $X_2, \dots, X_{\lceil n/2 \rceil - 1} $ generated by $I_2^2, \dots, I^2_{\lceil n/2 \rceil - 1}$ are linearly independent almost everywhere on $\mathcal P_n / E$. Since the submanifold $\mathcal P_n / E$ of $\mathcal P_n^E / E$ is locally defined by the equations $I_0^2 = 1$ and \eqref{eq:cc}, and all $I_j$'s and $E_j$'s Poisson commute, it follows that the vector fields $X_2, \dots, X_{\lceil n/2 \rceil - 1} $ are tangent to $\mathcal P_n / E$. Moreover, those vector fields preserve the invariants $E_1, \dots, E_n, I^2_2, \dots, I^2_{\lfloor n/2 \rfloor}$ (again, because $I_j$'s and $E_j$'s Poisson commute). So, the recutting action on $\mathcal P_n / E$ has $\lfloor 3n/2 \rfloor$ independent first integrals  $E_1, \dots, E_n, I^2_2, \dots, I^2_{\lfloor n/2 \rfloor}$ and a complementary number $\lceil n/2 \rceil - 2$ of commuting invariant vector fields $X_2, \dots, X_{\lceil n/2 \rceil - 1} $ which are also independent and tangent to joint level sets of first integrals. Thus, the first part of Theorem \ref{thm2} is proved. Furthermore, the second part now follows from the non-Hamiltonian version of the Arnold-Liouville theorem, see e.g. \cite{bogoyavlenskij1998extended}. Thus, Theorem \ref{thm2} is proved.
\end{proof}

\section*{Appendix: The braid relation}
\addcontentsline{toc}{section}{Appendix: The braid relation}
\setcounter{lemma}{0}
\renewcommand{\thesection}{A}
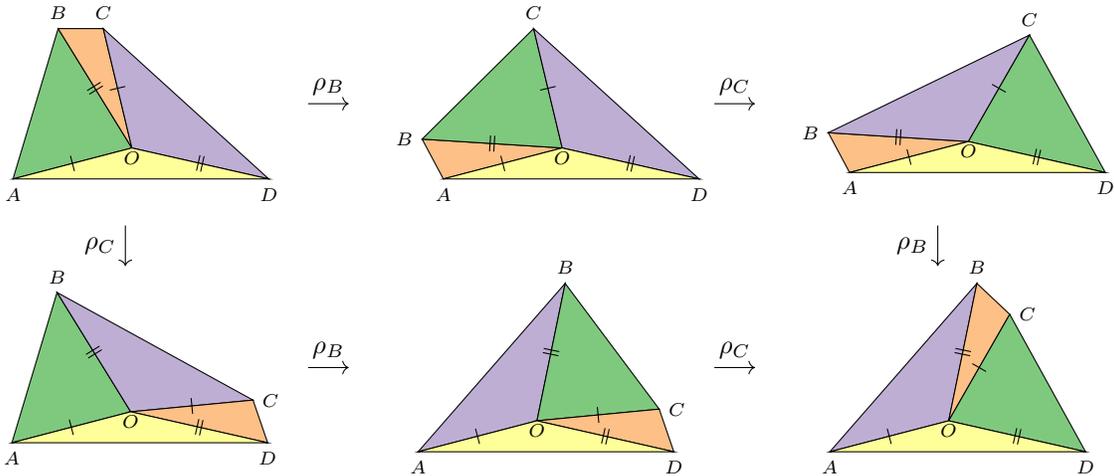
\begin{figure}[b!]
\centering
\begin{tikzpicture}[scale = 0.27]

\draw [->] (8,0) -- (10,0) node[midway, above] {$\rho_B$}; ;
\draw [->] (28,0) -- (30,0) node[midway, above] {$\rho_{C}$}; ;
\draw [->] (8,-13) -- (10,-13) node[midway, above] {$\rho_B$};
\draw [->] (28,-13) -- (30,-13) node[midway, above] {$\rho_{C}$}; ;
\draw [->] (-1,-6) -- (-1,-8) node[midway, left] {$\rho_{C}$}; ;
\draw [->] (39,-6) -- (39,-8) node[midway, left] {$\rho_{B}$}; ;
\node at (0,0)
{
\begin{tikzpicture}[scale = 0.2]
\coordinate [label=below:{$\scriptstyle{A}$}]  (A) at (0,-10) {};
\coordinate [label=above:{$\scriptstyle{B}$}]  (B) at (3,0) {};
\coordinate [label=above:{$\scriptstyle{C}$}](C) at (6,0) {};
\coordinate[label=below:{$\scriptstyle{D}$}] (D) at (17,-10) {};
\coordinate  (O) at (7.9,-7.94) {};
%\coordinate  [label=below:{$O$}] (Bp) at (-1.4118,-7.3529) {};
%\node () at (A) {};
%\coordinate[label=left:{$V_{i-1}$}] (B) at (0,1) {};
%\coordinate [label=left:{$V_{i}$}] (C) at (1,2) {};
%\coordinate [label=above:{$V_{i}'$}] (Cp) at (2,2) {};
%\coordinate [label=right:{$V_{i+1}$}] (D) at (3,1) {};
%\coordinate (E) at (2.5,0) {};
%\coordinate (F) at (1.5,1) {};
%\coordinate (G) at (1.5,2.5) {};
%\coordinate (H) at (1.5,0.5) {};
%\draw [](A) -- (B) -- (C) -- (D) -- (A);
%\draw (A) -- (O) -- (B) -- (O) -- (C) -- (O) -- (D);

\fill [color = color1, draw = black] (A) -- (O) -- (B) -- cycle ;
%\path  (A) -- (O) -- (B)  ;
\fill [color = color3, draw = black] (B) -- (O) -- (C) -- cycle;
\fill [color =  color2, draw = black] (C) -- (O) -- (D) -- cycle;
\fill [ color =  color4, draw = black] (A) -- (O) -- (D) -- cycle;
%\node at (barycentric cs:A=1,B=1,O=1) {$\scriptstyle\Delta_1$} ;
%\node at (barycentric cs:C=1,B=1,O=1) {$\scriptstyle\,\Delta_2$} ;
%\node at (barycentric cs:C=1,D=1,O=1) {$\scriptstyle\Delta_3$} ;
%\node [draw,circle,color=black, fill=black,inner sep=0pt,minimum size=3pt] () at (A) {};
%\draw [dashed] (B) -- (D);
%%\draw [dashed] (H) -- (G);
%\draw (B) -- (Cp) -- (D) ;
%\draw [] (A) -- (O);
%\draw [] (B) -- (O);
%\draw [] (C) -- (O);
%\draw [] (D) -- (O);

  \tkzMarkSegment[size = 3pt, pos=.5,mark=|](A,O);
  \tkzMarkSegment[size = 3pt, pos=.5,mark=|](O,C);
  \tkzMarkSegment[size = 3pt, pos=.5,mark=||](B,O);
  \tkzMarkSegment[size = 3pt, pos=.5,mark=||](O,D);
  %\node  [label = right:{$\scriptstyle{O}$}] () at (O) {};
\node    () at (7.9,-8.6) {$\scriptstyle{O}$};
\end{tikzpicture}
\quad
%$\longrightarrow$
};
\node at (0,-13)
{
\begin{tikzpicture}[scale = 0.2]

\coordinate [label=below:{$\scriptstyle{A}$}]  (A) at (0,-10) {};
\coordinate [label=above:{$\scriptstyle{B}$}]  (B) at (3,0) {};
\coordinate [label=right:{$\scriptstyle{C}$}](C) at (16.027, -7.1622) {};
\coordinate[label=below:{$\scriptstyle{D}$}] (D) at (17,-10) {};
\coordinate  (O) at (7.9,-7.94) {};
%\coordinate  [label=below:{$O$}] (Bp) at (-1.4118,-7.3529) {};
%\node () at (A) {};
%\coordinate[label=left:{$V_{i-1}$}] (B) at (0,1) {};
%\coordinate [label=left:{$V_{i}$}] (C) at (1,2) {};
%\coordinate [label=above:{$V_{i}'$}] (Cp) at (2,2) {};
%\coordinate [label=right:{$V_{i+1}$}] (D) at (3,1) {};
%\coordinate (E) at (2.5,0) {};
%\coordinate (F) at (1.5,1) {};
%\coordinate (G) at (1.5,2.5) {};
%\coordinate (H) at (1.5,0.5) {};
%\draw [](A) -- (B) -- (C) -- (D) -- (A);
%\draw (A) -- (O) -- (B) -- (O) -- (C) -- (O) -- (D);

\fill [color = color1, draw = black] (A) -- (O) -- (B) -- cycle ;
%\path  (A) -- (O) -- (B)  ;
\fill [color = color2, draw = black] (B) -- (O) -- (C) -- cycle;
\fill [color =  color3, draw = black] (C) -- (O) -- (D) -- cycle;
\fill [ color =  color4, draw = black] (A) -- (O) -- (D) -- cycle;
%\node at (barycentric cs:A=1,B=1,O=1) {$\scriptstyle\Delta_1$} ;
%\node at (barycentric cs:C=1,B=1,O=1) {$\scriptstyle\,\bar \Delta_3$} ;
%\node at (barycentric cs:C=1,D=1,O=1) {$\scriptstyle\bar \Delta_2$} ;
%\node [draw,circle,color=black, fill=black,inner sep=0pt,minimum size=3pt] () at (A) {};
%\draw [dashed] (B) -- (D);
%%\draw [dashed] (H) -- (G);
%\draw (B) -- (Cp) -- (D) ;
%\draw [] (A) -- (O);
%\draw [] (B) -- (O);
%\draw [] (C) -- (O);
%\draw [] (D) -- (O);

  \tkzMarkSegment[size = 3pt, pos=.5,mark=|](A,O);
  \tkzMarkSegment[size = 3pt, pos=.5,mark=|](O,C);
  \tkzMarkSegment[size = 3pt, pos=.5,mark=||](B,O);
  \tkzMarkSegment[size = 3pt, pos=.5,mark=||](O,D);
\node    () at (7.9,-8.6) {$\scriptstyle{O}$};
\end{tikzpicture}
\quad
%$\longrightarrow$
};
\node at (20,-13)
{
\begin{tikzpicture}[scale = 0.2]
\coordinate [label=below:{$\scriptstyle{A}$}]  (A) at (0,-10) {};
\coordinate [label=above:{$\scriptstyle{B}$}]  (B) at (9.7758,1.1998) {};
\coordinate [label=right:{$\scriptstyle{C}$}](C) at (16.027, -7.1622) {};
\coordinate[label=below:{$\scriptstyle{D}$}] (D) at (17,-10) {};
\coordinate  (O) at (7.9,-7.94) {};
%\coordinate  [label=below:{$O$}] (Bp) at (-1.4118,-7.3529) {};
%\node () at (A) {};
%\coordinate[label=left:{$V_{i-1}$}] (B) at (0,1) {};
%\coordinate [label=left:{$V_{i}$}] (C) at (1,2) {};
%\coordinate [label=above:{$V_{i}'$}] (Cp) at (2,2) {};
%\coordinate [label=right:{$V_{i+1}$}] (D) at (3,1) {};
%\coordinate (E) at (2.5,0) {};
%\coordinate (F) at (1.5,1) {};
%\coordinate (G) at (1.5,2.5) {};
%\coordinate (H) at (1.5,0.5) {};
%\draw [](A) -- (B) -- (C) -- (D) -- (A);
%\draw (A) -- (O) -- (B) -- (O) -- (C) -- (O) -- (D);

\fill [color = color2, draw = black] (A) -- (O) -- (B) -- cycle ;
%\path  (A) -- (O) -- (B)  ;
\fill [color = color1, draw = black] (B) -- (O) -- (C) -- cycle;
\fill [color =  color3, draw = black] (C) -- (O) -- (D) -- cycle;
\fill [ color =  color4, draw = black] (A) -- (O) -- (D) -- cycle;
%\node at (barycentric cs:A=1,B=1,O=1) {$\scriptstyle\Delta_3$} ;
%\node at (barycentric cs:C=1,B=1,O=1) {$\scriptstyle\,\bar \Delta_1$} ;
%\node at (barycentric cs:C=1,D=1,O=1) {$\scriptstyle\bar \Delta_2$} ;
%\node [draw,circle,color=black, fill=black,inner sep=0pt,minimum size=3pt] () at (A) {};
%\draw [dashed] (B) -- (D);
%%\draw [dashed] (H) -- (G);
%\draw (B) -- (Cp) -- (D) ;
%\draw [] (A) -- (O);
%\draw [] (B) -- (O);
%\draw [] (C) -- (O);
%\draw [] (D) -- (O);

  \tkzMarkSegment[size = 3pt, pos=.5,mark=|](A,O);
  \tkzMarkSegment[size = 3pt, pos=.5,mark=|](O,C);
  \tkzMarkSegment[size = 3pt, pos=.5,mark=||](B,O);
  \tkzMarkSegment[size = 3pt, pos=.5,mark=||](O,D);
\node    () at (7.9,-8.6) {$\scriptstyle{O}$};
\end{tikzpicture}
\quad
%$\longrightarrow$
};
\node at (20,0) {
\begin{tikzpicture}[scale = 0.2]
\coordinate [label=below:{$\scriptstyle{A}$}]  (A) at (0,-10) {};
%\coordinate [label=left:{${i}$}]  (B) at (3,0) {};
\coordinate [label=above:{$\scriptstyle{C}$}](C) at (6,0) {};
\coordinate[label=below:{$\scriptstyle{D}$}] (D) at (17,-10) {};
\coordinate   (O) at (7.9,-7.94) {};
\coordinate   [label=left:{$\scriptstyle{B}$}] (Bp) at (-1.4118,-7.3529) {};
%\node () at (A) {};
%\coordinate[label=left:{$V_{i-1}$}] (B) at (0,1) {};
%\coordinate [label=left:{$V_{i}$}] (C) at (1,2) {};
%\coordinate [label=above:{$V_{i}'$}] (Cp) at (2,2) {};
%\coordinate [label=right:{$V_{i+1}$}] (D) at (3,1) {};
%\coordinate (E) at (2.5,0) {};
%\coordinate (F) at (1.5,1) {};
%\coordinate (G) at (1.5,2.5) {};
%\coordinate (H) at (1.5,0.5) {};
%\draw [](A) -- (B) -- (C) -- (D) -- (A);
%\draw (A) -- (O) -- (B) -- (O) -- (C) -- (O) -- (D);

\fill [color = color3, draw = black] (A) -- (O) -- (Bp) -- cycle ;
%\path  (A) -- (O) -- (B)  ;
\fill [color = color1, draw = black] (Bp) -- (O) -- (C) -- cycle;
\fill [color =  color2, draw = black] (C) -- (O) -- (D) -- cycle;
\fill [ color =  color4, draw = black] (A) -- (O) -- (D) -- cycle;
%\node at (barycentric cs:A=1,Bp=1,O=1) {$\scriptstyle\bar \Delta_2$} ;
%\node at (barycentric cs:C=1,Bp=1,O=1) {$\scriptstyle\bar \Delta_1$} ;
%\node at (barycentric cs:C=1,D=1,O=1) {$\scriptstyle\Delta_3$} ;
%\node [draw,circle,color=black, fill=black,inner sep=0pt,minimum size=3pt] () at (A) {};
%\draw [dashed] (B) -- (D);
%%\draw [dashed] (H) -- (G);
%\draw (B) -- (Cp) -- (D) ;
%\draw [] (A) -- (O);
%\draw [] (B) -- (O);
%\draw [] (C) -- (O);
%\draw [] (D) -- (O);

  \tkzMarkSegment[size = 3pt, pos=.5,mark=|](A,O);
  \tkzMarkSegment[size = 3pt, pos=.5,mark=|](O,C);
  \tkzMarkSegment[size = 3pt, pos=.5,mark=||](Bp,O);
  \tkzMarkSegment[size = 3pt, pos=.5,mark=||](O,D);
\node    () at (7.9,-8.6) {$\scriptstyle{O}$};
\end{tikzpicture}
};
\node at (40,0) {
\begin{tikzpicture}[scale = 0.2]
\coordinate [label=below:{$\scriptstyle{A}$}]  (A) at (0,-10) {};
%\coordinate [label=left:{${i}$}]  (B) at (3,0) {};
\coordinate [label=above:{$\scriptstyle{C}$}](Cp) at (11.9599,-0.8569) {};
\coordinate[label=below:{$\scriptstyle{D}$}] (D) at (17,-10) {};
\coordinate   (O) at (7.9,-7.94) {};
\coordinate   [label=left:{$\scriptstyle{B}$}] (Bp) at (-1.4118,-7.3529) {};
%\node () at (A) {};
%\coordinate[label=left:{$V_{i-1}$}] (B) at (0,1) {};
%\coordinate [label=left:{$V_{i}$}] (C) at (1,2) {};
%\coordinate [label=above:{$V_{i}'$}] (Cp) at (2,2) {};
%\coordinate [label=right:{$V_{i+1}$}] (D) at (3,1) {};
%\coordinate (E) at (2.5,0) {};
%\coordinate (F) at (1.5,1) {};
%\coordinate (G) at (1.5,2.5) {};
%\coordinate (H) at (1.5,0.5) {};
%\draw [](A) -- (B) -- (C) -- (D) -- (A);
%\draw (A) -- (O) -- (B) -- (O) -- (C) -- (O) -- (D);

\fill [color = color3, draw = black] (A) -- (O) -- (Bp) -- cycle ;
%\path  (A) -- (O) -- (B)  ;
\fill [color = color2, draw = black] (Bp) -- (O) -- (Cp) -- cycle;
\fill [color =  color1, draw = black] (Cp) -- (O) -- (D) -- cycle;
\fill [ color =  color4, draw = black] (A) -- (O) -- (D) -- cycle;
%\node at (barycentric cs:A=1,Bp=1,O=1) {$\scriptstyle\bar \Delta_2$} ;
%\node at (barycentric cs:Cp=1,Bp=1,O=1) {$\scriptstyle\bar \Delta_3$} ;
%\node at (barycentric cs:Cp=1,D=1,O=1) {$\scriptstyle\Delta_1$} ;
%\node [draw,circle,color=black, fill=black,inner sep=0pt,minimum size=3pt] () at (A) {};
%\draw [dashed] (B) -- (D);
%%\draw [dashed] (H) -- (G);
%\draw (B) -- (Cp) -- (D) ;
%\draw [] (A) -- (O);
%\draw [] (B) -- (O);
%\draw [] (C) -- (O);
%\draw [] (D) -- (O);

  \tkzMarkSegment[size = 3pt, pos=.5,mark=|](A,O);
  \tkzMarkSegment[size = 3pt, pos=.5,mark=|](O,Cp);
  \tkzMarkSegment[size = 3pt, pos=.5,mark=||](Bp,O);
  \tkzMarkSegment[size = 3pt, pos=.5,mark=||](O,D);
\node    () at (7.9,-8.6) {$\scriptstyle{O}$};
\end{tikzpicture}
};
\node at (40,-13) {
\begin{tikzpicture}[scale = 0.2]
\coordinate [label=below:{$\scriptstyle{A}$}]  (A) at (0,-10) {};
%\coordinate [label=left:{${i}$}]  (B) at (3,0) {};
\coordinate [label=right:{$\scriptstyle{C}$}](Cp) at (11.9599,-0.8569) {};
\coordinate[label=below:{$\scriptstyle{D}$}] (D) at (17,-10) {};
\coordinate   (O) at (7.9,-7.94) {};
\coordinate   [label=above:{$\scriptstyle{B}$}] (Bp) at (9.7758,1.1998) {};
%\node () at (A) {};
%\coordinate[label=left:{$V_{i-1}$}] (B) at (0,1) {};
%\coordinate [label=left:{$V_{i}$}] (C) at (1,2) {};
%\coordinate [label=above:{$V_{i}'$}] (Cp) at (2,2) {};
%\coordinate [label=right:{$V_{i+1}$}] (D) at (3,1) {};
%\coordinate (E) at (2.5,0) {};
%\coordinate (F) at (1.5,1) {};
%\coordinate (G) at (1.5,2.5) {};
%\coordinate (H) at (1.5,0.5) {};
%\draw [](A) -- (B) -- (C) -- (D) -- (A);
%\draw (A) -- (O) -- (B) -- (O) -- (C) -- (O) -- (D);

\fill [color = color2, draw = black] (A) -- (O) -- (Bp) -- cycle ;
%\path  (A) -- (O) -- (B)  ;
\fill [color = color3, draw = black] (Bp) -- (O) -- (Cp) -- cycle;
\fill [color =  color1, draw = black] (Cp) -- (O) -- (D) -- cycle;
\fill [ color =  color4, draw = black] (A) -- (O) -- (D) -- cycle;
%\node at (barycentric cs:A=1,Bp=1,O=1) {$ \scriptstyle\Delta_3$} ;
%\node at (barycentric cs:Cp=1,Bp=1,O=1) {$ \scriptstyle\Delta_2$} ;
%\node at (barycentric cs:Cp=1,D=1,O=1) {$\scriptstyle\Delta_1$} ;
%\node [draw,circle,color=black, fill=black,inner sep=0pt,minimum size=3pt] () at (A) {};
%\draw [dashed] (B) -- (D);
%%\draw [dashed] (H) -- (G);
%\draw (B) -- (Cp) -- (D) ;
%\draw [] (A) -- (O);
%\draw [] (B) -- (O);
%\draw [] (C) -- (O);
%\draw [] (D) -- (O);

  \tkzMarkSegment[size = 3pt, pos=.5,mark=|](A,O);
  \tkzMarkSegment[size = 3pt, pos=.5,mark=|](O,Cp);
  \tkzMarkSegment[size = 3pt, pos=.5,mark=||](Bp,O);
  \tkzMarkSegment[size = 3pt, pos=.5,mark=||](O,D);
\node    () at (7.9,-8.6) {$\scriptstyle{O}$};
\end{tikzpicture}
};

\end{tikzpicture}
\caption{ The braid relation.}\label{FigA}
\end{figure}
Recuttings at adjacent vertices satisfy the braid relation. Adler \cite{Adler2} provides an algebraic argument. Here we give a geometric proof. 
\begin{proposition}\label{prop:braid}
Recuttings at adjacent vertices satisfy the braid relation
 $
\rho_i \rho_{i+1}\rho_i = \rho_{i+1} \rho_{i}\rho_{i+1}.
$
\end{proposition}
\begin{proof}
It suffices to show that for any quadrilateral $ABCD$ one has $\rho_B\rho_C\rho_B = \rho_C\rho_B\rho_C$. Moreover, it is sufficient to consider the case of a convex quadrilateral. The general case follows by analytic continuation. \par Let $ABCD$ be a convex quadrilateral, and let $O$ be the intersection point of perpendicular bisectors to its diagonals. Consider the triangles $AOB$, $BOC$, $COD$, $DOA$ shown in Figure \ref{FigA}. Observe that recutting at any vertex is equivalent to detaching two of those  triangles (namely those that are adjacent to the given vertex), and then attaching them back but switched and with opposite orientations. As a result, both transformations $\rho_B\rho_C\rho_B$ and $\rho_C\rho_B\rho_C$ boil down to cutting $ABCD$ into four triangles $AOB$, $BOC$, $COD$, $DOA$ and then gluing them back interchanging~${AOB}$ with ${COD}$. So, we indeed have $\rho_B\rho_C\rho_B = \rho_C\rho_B\rho_C$.\end{proof}
 \begin{remark}
 This argument shows that the intersection point of perpendicular bisectors of diagonals is invariant under recuttings of a quadrilateral. For more general polygons, a point with this property is known as the \textit{circumcenter of mass}. Consider an arbitrarily triangulation of a polygon. Place point masses at circumcenters of the triangles, with each mass being proportional to the area of the corresponding triangle. Then their center of mass is independent of the triangulation and is called the {circumcenter of mass} of the polygon \cite{TT, Akopyan}. Its invariance under recutting was observed in \cite{Adler}. For quadrilaterals, the circumcenter of mass is precisely the intersection point of perpendicular bisectors of diagonals \cite[Remark 3.3]{TT}.
 
 %Its invariance under recutting was observed in \cite{Adler}, while the name
 
  \end{remark}
  \begin{remark}\label{rem:quad}
  Since recutting of a quadrilateral amounts to switching colored triangles in Figure \ref{FigA}, there are only finitely many (isometry classes of) quads that can be obtained from a given one by means of a sequence of recuttings. 
  \end{remark}
\bibliographystyle{plain}
\bibliography{recut.bib}

\end{document}